\documentclass[a4paper,onecolumn,11pt,unpublished]{quantumarticle}

\pdfoutput=1
\usepackage[utf8]{inputenc}
\usepackage[english]{babel}
\usepackage[T1]{fontenc}
\usepackage{amsmath}
\usepackage{amssymb}
\usepackage{physics}

\usepackage{hyperref}
\usepackage{colortbl}
\usepackage[table]{xcolor}
\usepackage{booktabs}
\usepackage{enumerate}
\usepackage{enumitem}
\usepackage{caption}
\usepackage{subcaption}
\usepackage{amsthm}
\usepackage{caption}
\usepackage{subcaption}% <–– add this
\usepackage{selinput}
\SelectInputMappings{
  adieresis={ä}, % etc.
  acute={´},
}
\usepackage{amsmath}  % if not already loaded

\usepackage[english]{babel}
\usepackage[T1]{fontenc}
\usepackage{lipsum}
\usepackage{mdframed}
\usepackage{booktabs}
\usepackage{adjustbox}
\usepackage{longtable}
\usepackage{float}

\usepackage[
  backend=biber,
  style=numeric,
  sorting=none      % or 'nyt' if you want author–year sorting
]{biblatex}
\usepackage[table]{xcolor}
\usepackage{listings}
\usepackage{siunitx}
\usepackage{multicol}
\usepackage{pifont}  
\usepackage{newunicodechar}
\newunicodechar{✖}{\ding{55}}
\usepackage[margin=1in]{geometry}   % Adjust page margins
\usepackage{setspace}              % For line spacing
\usepackage{bm}
\usepackage{amsmath,amssymb,amsthm}             % Theorem and proof environments
\usepackage{tikz}            % TikZ is required by quanti
\usepackage{quantikz} 
\usepackage{enumitem}% For drawing quantum circuits
\usepackage{graphicx} 
\usepackage{algpseudocode}
\usepackage{algorithm}
\usepackage{xcolor}
\usepackage[most]{tcolorbox}
\usepackage{pgfplots}
\pgfplotsset{compat=1.17}
\newcommand{\poly}[1]{\mathrm{poly}\bigl(#1\bigr)}
\usepackage{booktabs}

\usepackage{float}                 % For [H] float placement
\usepackage{hyperref}              % Clickable links/references
\usepackage{url}

\usepackage{xcolor}                % For colored text (optional)
\usepackage{lscape}                % For landscape pages if needed
\usepackage{caption}
\usepackage{subcaption}
\usepackage{etoolbox}
\newtheorem{theorem}{Theorem}
\newtheorem{lemma}[theorem]{Lemma}
\newtheorem{proposition}[theorem]{Proposition}
\newtheorem{corollary}[theorem]{Corollary}
\newtheorem{definition}[theorem]{Definition}

\theoremstyle{remark}
\newtheorem{remark}[theorem]{Remark}

\newcommand{\OH}{\mathcal{H}_{\mathrm{OH}}}
\DeclareUnicodeCharacter{0301}{}

\makeatletter
\g@addto@macro\normalsize{%
    \abovedisplayskip 3pt plus 1pt minus 1pt%
    \abovedisplayshortskip 3pt plus 1pt minus 1pt%
    \belowdisplayskip 3pt plus 1pt minus 1pt%
    \belowdisplayshortskip 3pt plus 1pt minus 1pt%
}
\makeatother

\def\BibTeX{{\rm B\kern-.05em{\sc i\kern-.025em b}\kern-.08em
    T\kern-.1667em\lower.7ex\hbox{E}\kern-.125emX}}

\author{Chinonso Onah}
\affiliation{Volkswagen AG, Berliner Ring 2, 38440 Wolfsburg, Germany}
\affiliation{Department of Physics, RWTH Aachen University, 52056 Aachen, Germany}

\author{Kristel Michielsen}
\affiliation{Forschungszentrum Jülich, Germany}
\affiliation{Universit\"at zu K\"oln, 50923 K\"oln, Germany}

\title{Polynomial Time Quantum Approximation Schemes for Constrained Optimization}

\begin{document}
\maketitle

\begin{abstract}
When does a noisy quantum sampler enable an end-to-end polynomial-time algorithm with performance guarantees? Building on the instance-dependent finite-depth, finite-shot guarantees in the Constraint-Enhanced QAOA (CE--QAOA) studied  in Ref.~\cite{onahfinite}, we answer this question for a complete quantum--classical pipeline executed on near-term quantum devices. The analysis of Ref.~\cite{onahfinite}
provides conditions under which the ideal CE--QAOA sampler can assign
inverse-polynomial probability
\(
p_0=\Omega(n^{-k})
\)
to the set of optimal solutions. Here, we extend this analysis
and show that, with independent-shot amplification, polynomial-time
feasibility repair, and exact scoring, the resulting pipeline recovers
an optimum with high probability in polynomial runtime. It
therefore gives a quantum-powered fully polynomial randomized
approximation scheme, which we call an
\textbf{FPRAS\textsubscript{$q$}}.

We next show that this guarantee can survive device noise within an instance-dependent window. For effective circuit depth
\(G(n,p)=\Theta(pn)\), retaining a \(G(n,p)^{-1}\) fraction of the ideal
optimal mass is sufficient to recover an optimum with probability at
least \(1-\delta\) using
\[
\mathcal{O}\!\left(pn^{k+1}\log(1/\delta)\right)
\]
shots. Outside this exact-hit window, deterministic
repair continues to guarantee feasibility and yields an
instance-dependent \((1+\varepsilon)\)-approximation scheme whenever
the repair-induced objective inflation is controlled. We term this the 
\textbf{ $\varepsilon$-Polynomial-Time Quantum Approximation Scheme ($\varepsilon-$PTQAS)}. 

We also place our hybrid algorithm within the Chen--Cotler--Huang--Li oracle
model~\cite{ChenCotlerHuangLi2023NISQ} and locate the source of the
quantum--classical separation. On any
NP-hard kernel-admissible promise family, no polynomial-time classical
sampler can reproduce the inverse-polynomial optimal overlap of NP-HQ
unless
\(
\mathrm{NP}\subseteq\mathrm{BPP}
\),
even when granted the same classical repair and perfect access to the
constraint structure. The separation therefore lies in generating the
sampling distribution, not in classical processing

Finally, we introduce \textbf{Heavy-Hitter QAOA (HH-QAOA)}, a geometry-informed filtering refinement that  preserves these
conditional guarantees while reducing the retained candidate set from
\(\mathcal{O}(n^{k+1})\) to \(\mathcal{O}(n^k)\) and the classical
post-processing cost by one power of \(n\). We demonstrate the complete
pipeline on 127-qubit IBM Eagle-r3 processors using \textsc{QOptlib}~\cite{Osaba2024Qoptlib}
travelling-salesman instances with up to 100 logical variables,
matching or improving the published reference tours in every tested
case.
\end{abstract}

\section{Introduction}
\label{sec:intro}

Constrained combinatorial optimization covers a broad class of scientific and engineering problems in which the objective is evaluated over configurations satisfying hard constraints\cite{PapadimitriouSteiglitz1982,Lucas2014Ising}. Representative examples include the Travelling
Salesman Problem (TSP/ATSP)~\cite{Lawler1985TSP,GareyJohnson1979}, the
Quadratic Assignment Problem (QAP)~\cite{Koopmans1957QAP,Loiola2007QAPSurvey},
capacitated vehicle routing~\cite{TothVigo2014VRP}, generalized assignment
and multiple-knapsack problems~\cite{SahniGonzalez1976GAP,MartelloToth1990Knapsack},
\(k\)-dimensional matching for \(k\geq 3\)~\cite{Karp1972,garey1979computers},
shared-transportation problems~\cite{onah2025waas,hernandez2026quantum,onah2026quantumclassical}, and job-shop and flow-shop scheduling problems~\cite{GareyJohnsonSethi1976,LenstraRinnooyKan1977,hernandez2026quantum}.
Their admissible configurations are organized by assignment, matching,
routing, capacity, ordering, and scheduling constraints\cite{Lucas2014Ising}. The same structural motifs appear in polynomial-time solvable problems.
Linear assignment is solved exactly by the Hungarian
algorithm~\cite{Kuhn1955Hungarian,Munkres1957Hungarian}, and maximum
matching admits polynomial-time algorithms based on Edmonds' blossom
construction~\cite{Edmonds1965Blossom}. Assignment matrices, matching
relations, and one-hot or k-hot cardinality constraints therefore provide a common structural language across tractable and NP-hard optimization problems. One-hot encodings represent cardinality constraints by assigning one excitation to each logical block while relations among blocks enforce the remaining local and global constraints. The corresponding feasible configurations form a subset of the already cardinality-constrained computational basis, and the quantum sampler must place sufficient probability on this subset to produce candidates for downstream processing. 

However, concentration of probability mass on the feasible subspace is only the first step in constrained quantum optimization. To obtain an end-to-end algorithm, one must also identify and score candidate solutions efficiently from the obtained samples. On real devices, quantum noise degrades feasibility further through noise-induced constraint violations, so one must additionally repair corrupted samples before scoring. This paper develops that full algorithmic layer around the recently introduced  Onah--Firt--Michielsen (OFM) kernel of Ref.~\cite{onahce,onahseperating} and shows how to exploit the same kernel structure classically, through deterministic checking, repair, and scoring, so that the quantum and classical stages operate as a single optimization pipeline. The CE-QAOA sampler already exploits these structures and symmetries in the quantum dynamics, leading to several guarantees studied in Refs.~\cite{ onahfinite, onahseperating,onahfund,onah2026cvrp} and reviewed in Appendix~\ref{sec:ce-qaoa}. The associated polynomial-time hybrid quantum--classical (PHQC)
algorithm was introduced in Ref.~\cite{onahce}. In that construction,
quantum samples generated by CE--QAOA are passed to a classical checker
for feasibility verification and scoring, while infeasible samples are
discarded. Here we propose a noise-aware extension called the Noisy Polytime Hybrid Quantum-Classical (NP-HQ) algorithm. The post-processor first checks feasibility, applies a problem-structured repair map when noise corrupts the one-hot patterns, evaluates the objective, and returns the best repaired candidate.

Our main results arise by exploiting the structures and symmetries of
the OFM kernel at both stages of the pipeline: through the
symmetry-powered quantum dynamics of CE--QAOA and through
problem-dependent repair maps applied classically to quantum-generated
samples. We show that, in the low-noise regime where the optimal solution retains
inverse-polynomial sampling weight, NP-HQ is an exact-hit fully polynomial randomized
approximation scheme,
FPRAS\textsubscript{$q$} with polynomial shot complexity and polynomial
classical runtime. Beyond that regime, deterministic repair still supports an instance-dependent
\((1+\varepsilon)\)-approximation scheme whenever its objective inflation is
controlled. A second contribution is a noise-to-runtime translation. We exploit
a Lipschitz-type total-variation bound to convert device-level noise directly
into a lower bound on post-noise optimal mass and hence into an explicit
polynomial shot budget. 

A further contribution is a conditional quantum--classical separation
for the resulting sampling pipeline. We show that the complete NP-HQ pipeline fits the Chen--Cotler--Huang--Li NISQ oracle model\cite{ChenCotlerHuangLi2023NISQ}. We compare NP-HQ with classical
polynomial-time strategies allowed to sample arbitrary efficiently
samplable distributions, apply the same classical repair map, and score
every repaired candidate. If such a classical strategy achieved
inverse-polynomial optimal overlap uniformly on an NP-hard
kernel-admissible promise family, standard success amplification would
place the corresponding promise search problem in
\(\mathrm{BPP}\). This conclusion remains unchanged when the classical
sampler is granted perfect access to row-one-hot strings or uniformly
random permutations. The computational distinction is therefore in the generation of
the optimum-biased sampling distribution. Finally, we introduce Heavy-Hitter QAOA (HH-QAOA), a refined post-processing strategy that retains only those bit-strings whose empirical frequencies exceed a prescribed threshold. Focusing on 
a small thresholded candidate set reduces post-processing cost 
substantially, and in the strongest regimes where additional problem structure present in the quantum-generated samples is utilized, the downstream checking, repair, and scoring cost depends on the size of the retained heavy-hitter set, leading to orders-of-magnitude reductions in the classical postprocessing time. We demonstrate the complete NP-HQ pipeline on 127-qubit IBM Eagle-r3
processors using \textsc{QOptlib}\cite{Osaba2024Qoptlib} travelling-salesman instances with up to 100 logical variables, matching or improving reference
tours in every tested case.

\subsection{Related Work}
\label{sec:related}

Post-selection and filtering heuristics appear widely in variational and
sampling-based quantum algorithms. Common strategies include discarding
low-quality samples, reweighting outcomes, and applying symmetry checks
to improve empirical performance%
~\cite{BonetMonroig2018SymmetryVerification,
Sagastizabal2019SymmetryVerification}. A substantial body of NISQ work
also addresses device noise through zero-noise extrapolation, symmetry
verification, readout mitigation, and regression-based error
mitigation%
~\cite{Temme2017ErrorMitigation,Endo2018PracticalQEM,
BonetMonroig2018SymmetryVerification,Bravyi2021MeasurementMitigation,
Czarnik2021CDR}. These methods primarily target expectation-value
estimation, whereas the present work concerns the finite-shot recovery
of optimal bitstrings, for which existing error-mitigation tools
are of limited utility.

The \emph{quantum alternating-operator} paradigm uses mixing
operators that preserve hard constraints, thereby confining the
 evolution to structured subspaces and yielding several reported
performance gains%
~\cite{Hadfield2019AOA,Fuchs2022ConstrainedMixers,sawaya2023encoding,
tsvelikhovskiy2024symmetries,onahce, onahfund,onahfinite,onahseperating}. In implementation, however,
constraint-preserving mixer Hamiltonians may contain noncommuting terms,
so their Trotterized realizations need not reproduce the ideal
subspace-preserving dynamics exactly~\cite{awasthi2026constraintpreservingxymixerstrotterized,onahseperating}. Device noise can further induce
leakage from the encoded subspace, motivating explicit noise-handling
mechanisms for constrained quantum optimization. Recent hardware studies
with XY-type mixers nevertheless show that constrained dynamics can
produce significant gains on near-term devices when these implementation
and noise obstacles are addressed%
~\cite{Niroula2022ExtractiveSummarization,Maciejewski_2025,hadfield2026ndar}. A related line of work exploits quantum noise using physical insights
from gauge transformations to improve the performance of hybrid
quantum algorithms%
~\cite{Maciejewski_2025,hadfield2026ndar,Maciejewski2026}.

Classical post-processing has also been incorporated directly into
quantum optimization pipelines. Greedy post-processing has been shown
to increase the approximation quality of QAOA samples%
~\cite{Caha2022TwistedHybrid}, while post-processing variational
algorithms can repair infeasible outputs produced by annealers or
gate-based devices%
~\cite{Shirai2024Postprocessing}. These works establish the algorithmic
value of classical processing after quantum sampling. The present
construction goes further by incorporating deterministic checking,
feasibility restoration, and exact scoring into a finite-shot
complexity analysis that tracks the optimal sampling mass through
device noise and converts it into an explicit end-to-end success
probability and runtime.

Heavy-hitter and frequent-item extraction originated in the classical
data-stream literature, where thresholded frequency summaries identify
items carrying a prescribed fraction of the observed mass%
~\cite{MisraGries1982,CormodeMuthukrishnan2005}. We apply this principle
to finite-shot quantum output histograms. The classical heavy-hitter
bound supplies a distribution-independent candidate-set reduction,
while the CE--QAOA optimal-mass bound and its device-noise window permit
a sharper quantum-informed threshold. The resulting refinement is
therefore tied to the structure and noise stability of the quantum
sampling distribution rather than generic frequency estimation.

 Classical approximation schemes for
structured metric optimization exploit substantial geometric or
combinatorial structure%
~\cite{Arora1998PTAS_TSP}. In the same spirit, we repurpose
polynomial-time assignment and matching algorithms%
~\cite{Kuhn1955Hungarian,Edmonds1965Blossom}
as deterministic feasibility-restoration primitives for measured
quantum samples.  From a complexity-theoretic perspective, our construction fits within the hybrid
model of noisy quantum computation studied by
Chen, Cotler, Huang, and Li %
~\cite{ChenCotlerHuangLi2023NISQ}. To the best of our knowledge, no previous work combines finite-shot
optimal recovery, deterministic feasibility restoration, an explicit
finite-noise lower bound on the recovered optimal mass, polynomial shot
complexity, and polynomial post-processing runtime into a single
end-to-end optimization guarantee under device noise. The hardware
experiments in Section~\ref{sec:experiment} instantiate the complete
sampling--checking--repair--scoring pipeline on contemporary
superconducting processors.

% The resulting guarantees are controlled by the optimal mass assigned by the quantum sampler before repair, its monotonicity under deterministic repair, and its stability under device noise. A total-variation
% continuity bound converts the surviving noisy optimal mass into an
% explicit finite-shot budget.

% ============================
% ============================
\section{Polynomial-Time Quantum Approximation Schemes}
\label{sec:ptqas}
% ============================

The approximation schemes developed below inherit their quantum
sampling guarantees from the CE--QAOA analyses of
Refs.~\cite{onahfinite,onahseperating,onahfund}. To make the assumptions
entering those guarantees explicit, we first state the
Onah--Firt--Michielsen (OFM) kernel in full. The kernel provides the
common interface between the quantum sampler and the classical
checker--repair--scoring stage: it fixes the encoded manifold, the
residual constraint Hamiltonians, their symmetries, the compatible
mixer, and the initial state. The subsequent complexity results apply
to problem families admitting this interface together with the stated
instance-dependent lower bound on optimal sampling mass.

%------------------------------------------------------------
\subsection{The OFM kernel and its assignment specialization}
\label{sec:ofm-kernel}
%------------------------------------------------------------

%------------------------------------------------------------

\begin{definition}[The Onah--Firt--Michielsen (OFM) kernel\cite{onahce,onahseperating}]
\label{def:kernel-requirement}
An optimization instance belongs to the
Onah--Firt--Michielsen kernel if there exist
\(m,n_{\mathrm{loc}}\in\mathbb{N}\) and a one-hot encoder such that the
quantum dynamics is initialized in the fixed-Hamming-weight space
\[
\OH
:=
(\mathcal{H}_1)^{\otimes m},
\qquad
\mathcal{H}_1
:=
\operatorname{span}
\{\ket{e_1},\ldots,\ket{e_{n_{\mathrm{loc}}}}\},
\]
with one excitation in each block. Its computational-basis states are
identified with words
\[
x=(x_1,\ldots,x_m)\in
X:=[n_{\mathrm{loc}}]^m,
\qquad
\ket{x}
=
\ket{e_{x_1}}\otimes\cdots\otimes\ket{e_{x_m}}.
\]
Thus,
\[
|X|=\dim(\OH)=n_{\mathrm{loc}}^m.
\]

In this encoded basis, the diagonal problem Hamiltonian decomposes as
\[
H_C
=
H_{\mathrm{pen}}
+
H_{\mathrm{obj}}
+
H_{\mathrm{pen}}',
\]
where \(H_{\mathrm{obj}}\) encodes the objective,
\(H_{\mathrm{pen}}\) contains the symmetry-preserving residual
constraint penalties, and \(H_{\mathrm{pen}}'\) contains additional
diagonal penalties that need not preserve the full kernel symmetry.
The following requirements hold.

\begin{enumerate}
\item[(a)] \emph{Penalty structure.}
The Hamiltonian \(H_{\mathrm{pen}}\) is a sum of squared affine
one-hot, degree, or capacity residuals, optionally together with linear
forbids, with integer coefficients bounded polynomially in the encoded
instance size. Consequently,
\[
\operatorname{spec}(H_{\mathrm{pen}})
\subseteq
\{0,1,\ldots,t_{\max}\},
\qquad
t_{\max}
=
\operatorname{poly}(m,n_{\mathrm{loc}}).
\]

\item[(b)] \emph{Pattern symmetry.}
The symmetry-preserving penalty \(H_{\mathrm{pen}}\) is invariant under
block permutations \(S_m\) and global symbol relabelings
\(S_{n_{\mathrm{loc}}}\). It therefore organizes the encoded space into
structured penalty levels
\[
L_t
:=
\left\{
x\in X:
\bra{x}H_{\mathrm{pen}}\ket{x}=t
\right\}.
\]
The term \(H_{\mathrm{pen}}'\) collects problem-dependent penalty
contributions that break some or all of these symmetries.

\item[(c)] \emph{Mixer and initial state.}
The canonical one-hot mixer is the normalized block-local \(XY\)
Hamiltonian
\begin{equation}
\label{eq:mixerm}
\widetilde H_{XY}^{(b)}
=
\frac{1}{n_{\mathrm{loc}}-1}
\sum_{1\leq u<v\leq n_{\mathrm{loc}}}
\left(
X_u^{(b)}X_v^{(b)}
+
Y_u^{(b)}Y_v^{(b)}
\right),
\end{equation}
which preserves \(\mathcal{H}_1\) and satisfies
\[
\left\|
\widetilde H_{XY}^{(b)}
\right\|
=
O(1)
\]
on the one-excitation sector. The corresponding initial state is
\[
\ket{s_0}
=
\ket{s_{\mathrm{blk}}}^{\otimes m},
\qquad
\ket{s_{\mathrm{blk}}}
=
\frac{1}{\sqrt{n_{\mathrm{loc}}}}
\sum_{j=1}^{n_{\mathrm{loc}}}\ket{e_j}.
\]
\end{enumerate}
\end{definition}
Definition~\ref{def:kernel-requirement} provides a common interface
from which problem-structure-aware operations can be systematically
derived for many well-studied problems, including the TSP, QAP, CVRP,
shared transportation problems, generalized assignment and
multiple-knapsack problems, and \(k\)-dimensional matching for
\(k\geq 3\)%
~\cite{onah2025waas,onah2026quantumclassical,onah2026cvrp,
Lucas2014Ising,GareyJohnson1979,Karp1972,onahce}.

The kernel separates two forms of constraint handling. Local
cardinality constraints are enforced by the encoder and are preserved
 by the mixer, so the ideal quantum dynamics remains inside
\(X\). Residual global constraints are retained in the diagonal
penalty Hamiltonians and define a generally smaller globally feasible
set
\[
\mathcal{F}\subseteq X.
\]
The same residual constraint structure is subsequently available to
the classical checker and repair map. For the assignment and travelling-salesman specialization,
\[
m=n_{\mathrm{loc}}=n,
\qquad
X=[n]^n,
\qquad
|X|=n^n,
\]
and the encoding uses \(n^2\) physical qubits. Each word
\[
x=(x_1,\ldots,x_n)\in[n]^n
\]
selects one symbol in every one-hot block. Global feasibility further
requires every symbol to occur exactly once, so that
\[
\mathcal{F}
=
\left\{
x\in[n]^n:
x_a\neq x_b
\text{ for all }a\neq b
\right\}
\cong S_n.
\]
Thus the mixer explores the full encoded assignment space \(X\), while
the residual penalties distinguish the permutation sector
\(\mathcal{F}\subset X\).

%------------------------------------------------------------
\subsection{CE--QAOA sampling distribution}
\label{sec:ce-qaoa-setup}
%------------------------------------------------------------

Define the total mixer Hamiltonian and its unitary by
\[
H_M
:=
\sum_{b=1}^{m}\widetilde H_{XY}^{(b)},
\qquad
U_M(\beta)
:=
e^{-i\beta H_M}
=
\bigotimes_{b=1}^{m}
\exp\!\left(
-i\beta\widetilde H_{XY}^{(b)}
\right).
\]
Because the block terms act on disjoint registers,
\(U_M(\beta)\) preserves \(\OH\). Since \(H_C\) is diagonal in the
computational basis, the problem unitary
\[
U_C(\gamma):=e^{-i\gamma H_C}
\]
also preserves \(\OH\). A depth-\(p\) CE--QAOA circuit therefore
prepares
\begin{equation}
\label{eq:app-ce-state}
\ket{\psi_p(\boldsymbol{\beta},\boldsymbol{\gamma})}
=
\left(
\prod_{r=1}^{p}
U_M(\beta_r)U_C(\gamma_r)
\right)
\ket{s_0}
\in\OH.
\end{equation}

Its ideal output distribution is

\[
q_0(x)
:=
\left|
\langle x|
\psi_p(\boldsymbol{\beta},\boldsymbol{\gamma})
\rangle
\right|^2,
\qquad
x\in X.
\]

Let
\[
C(x):=\bra{x}H_{\mathrm{obj}}\ket{x}
\]
denote the true objective evaluated on a feasible configuration, and
define
\[
C^\star
:=
\min_{x\in\mathcal{F}}C(x),
\qquad
\Omega^\star
:=
\left\{
x\in\mathcal{F}:C(x)=C^\star
\right\}.
\]
The ideal optimal sampling mass is then
\[
q_0
:=
q_0(\Omega^\star)
=
\sum_{x\in\Omega^\star}q_0(x).
\]
The exact-hit and approximation guarantees developed below are
conditioned on the instance-dependent regime
\[
q_0
\geq
c\,n^{-k}
\]
for fixed constants \(c>0\) and \(k\geq0\). The remaining
analysis is to determine when this optimal mass survives device noise,
how many circuit repetitions are then sufficient for recovery, and how
deterministic checking, repair, and scoring complete the
polynomial-time hybrid pipeline.

\subsection{Background to Performance Guarantees}
\label{sec:beyond-classical}

The problem-dependent assumption rests on the lattice normalization of the problem Hamiltonian which proceeds as follows: Fix a base angle $\gamma$ and define the \emph{wrapped phase} of each
basis string:
\begin{equation}\label{eq:app-wrapped-phase}
  \theta(z) \;:=\; \gamma\,E(z) \pmod{2\pi}.
\end{equation}
Let $ H_C\ket{x}=E(x)\ket{x},\qquad x\in X$ so that $\theta^\star := \gamma\,E^\star \pmod{2\pi}$ is the common wrapped
phase of the optimal energy level. Define the \emph{wrapped phase
separation}:
\begin{equation}\label{eq:app-phase-gap}
  \delta
  \;:=\;
  \min_{y \notin \Omega^\star}
    \operatorname{dist}_{\mathbb{T}}\!\bigl(\theta(y),\,\theta^\star\bigr),
\end{equation}
where
$\operatorname{dist}_{\mathbb{T}}(\phi,\varphi)
  := \min_{k \in \mathbb{Z}} |\phi - \varphi + 2\pi k|
  \in [0,\pi]$.
The condition $\delta > 0$ means that the optimal phase is separated from
all suboptimal phases on the circle.

\begin{remark}[When $\delta > 0$ holds]
\label{rem:app-phase-gap}
The wrapped phase separation $\delta$ depends on the cost spectrum
$\{E(z)\}$ and the angle $\gamma$. For generic $\gamma$ (i.e., outside a
measure-zero set), distinct energy levels map to distinct phases, so
$\delta > 0$. Moreover, if the energy spectrum has integer or
rational-multiple structure (as is common in combinatorial problems with
integer-valued costs), then a suitable $\gamma$ can be chosen to maximize
$\delta$. As discussed in Ref.~\cite{onahfinite}, one may also relax the
exact phase-gap assumption by averaging over a grid of $\gamma$ values
when several near-optimal levels cluster in phase; the resulting bound
degrades gracefully.
\end{remark}

The finite-depth, finite-shot guarantees on which the remainder of the
analysis rests arise from three analytic ingredients, whose interplay is
reviewed in detail in Appendix~\ref{app:fejer-mechanism}. Here we recall
the key ingredients.

\paragraph{Ingredient 1. Mixer envelope.}
The block-local XY mixer, viewed through a cost-basis dephasing channel, 
induces a doubly stochastic Markov chain on the diagonal of the density 
operator. Starting from the uniform one-hot product state $\ket{s_0}$, 
the iterated transition gives a nonnegative, normalized mixer envelope
\begin{equation}\label{eq:main-Wp}
  W_p(z;\boldsymbol\beta)
  \;=\;
  \bigl[M_{\beta_p}\cdots M_{\beta_1}\,v^{(0)}\bigr](z),
  \qquad
  v^{(0)}(z) = |\braket{z}{s_0}|^2,
\end{equation}
where each $M_\beta(z|y) = |\bra{z}U_M(\beta)\ket{y}|^2$ is 
unistochastic (see Eq.~\eqref{eq:app-Wp} in 
Appendix~\ref{app:fejer-mechanism}). This envelope quantifies how the 
mixer spreads probability across the encoded manifold. The key quantity for the success bound is the 
optimal-set envelope weight
\begin{equation}\label{eq:main-Cbeta}
  C_\beta
  \;=\;
  \sum_{x\in\Omega^\star} W_p(x;\boldsymbol\beta),
\end{equation}
the total mixer-envelope mass on the set $\Omega^\star$ of  
optimal feasible strings.

\paragraph{Ingredient 2. Fej\'{e}r phase filter.}
The diagonal cost operator $H_C$ assigns a wrapped phase 
$\theta(z) = \gamma E(z) \pmod{2\pi}$ to each basis string, where 
$\gamma$ is the cost-layer angle. The Fej\'{e}r kernel of order $p$ is 
the nonnegative trigonometric polynomial
\begin{equation}\label{eq:app-fejer}
  F_p(\vartheta)
  \;=\;
  \frac{1}{p+1}
  \left|
    \sum_{r=0}^{p} e^{ir\vartheta}
  \right|^2
  \;=\;
  \frac{1}{p+1}
  \left(
    \frac{\sin\!\bigl(\frac{(p+1)\vartheta}{2}\bigr)}
         {\sin\!\bigl(\frac{\vartheta}{2}\bigr)}
  \right)^{\!2},
\end{equation}
which peaks at $\vartheta = 0$ with value $F_p(0) = p+1$ and decays 
away from the peak. Strings whose cost phase lies close to the optimal 
phase $\theta^\star$ receive amplification of order $p+1$, while strings 
separated by a phase gap $\delta > 0$ are suppressed by 
$M_p(\delta) \leq [(p+1)\sin^2(\delta/2)]^{-1}$.

\paragraph{Ingredient 3. Factorized reference distribution.}
In a dephased reference model, a normalized factorized reference 
distribution emerges as
\begin{equation}\label{eq:app-factorization}
  \Pr_p^{\mathrm{ref}}[z]
  \;=\;
  \frac{
    W_p(z;\boldsymbol\beta)\;
    F_p\!\bigl(\theta(z) - \theta^\star\bigr)
  }{
    \displaystyle
    \sum_{y \in X}
      W_p(y;\boldsymbol\beta)\;
      F_p\!\bigl(\theta(y) - \theta^\star\bigr)
  }.
\end{equation}
This factorized distribution leads to a dimension-free lower bound on 
the optimal-set probability by combining Eq. \ref{eq:main-Cbeta} and Eq. \ref{eq:app-fejer}(Theorem~\ref{thm:app-fejer-success})
\begin{equation}\label{eq:main-q0}
  q_0
  \;\geq\;
  \frac{(p+1)\,C_\beta}
       {(p+1)\,C_\beta \;+\; M_p(\delta)\,(1-C_\beta)},
\end{equation}
which depends only on the filter order $p$, the phase gap $\delta$, and 
the optimal-set envelope weight $C_\beta$. The detailed derivation is 
given in Appendix~\ref{app:fejer-mechanism}. When evaluated on a problem family for which the instance-dependent
envelope and phase conditions detailed in
Appendix~\ref{app:consequences} hold, this joint success probability law yields the inverse-polynomial optimal-mass premise used throughout the downstream analysis:
\begin{equation}
\label{eq:main-q0p}
  q_0
  \;\geq\;
  \frac{(p+1)c_0 n^{-a}}
       {(p+1)c_0 n^{-a}
        + \bigl[(p+1)c_1\bigr]^{-1}}
  \;=\;
  \Omega(n^{-a}).
\end{equation}

The load-bearing geometric quantity in the finite-depth Fej\'{e}r
mechanism of Ref.~\cite{onahfinite} is the optimal-set mixer-envelope
weight \(C_\beta\) defined in Eq.~\eqref{eq:main-Cbeta}. The recent
geometry--interference analysis of Ref.~\cite{onahseperating} suggests
that the corresponding optimal-set transport weight need not be small
under coherent implementation. It also establishes a complementary
logarithmic-depth route to inverse-polynomial success guarantees under
the same normalized OFM transport structure and lattice normalization.
No general pointwise ordering between coherent and dephased output
probabilities is assumed here; the dephased reference law is compared
with the coherent circuit only under the phase-alignment conditions of
Ref.~\cite{onahseperating}. The logical flow from kernel structure to
end-to-end hybrid guarantees is visualized in the diagram below.

\begin{center}
\label{fig:main}
\begin{tikzpicture}[
  >=Stealth, font=\small,
  box/.style={draw, rounded corners=4pt, thick, align=center,
              inner sep=5pt, minimum width=3.2cm},
  arr/.style={-Stealth, thick}
]
\node[box] (K) at (0,0)
  {OFM kernel\\(Def.~\ref{def:kernel-requirement})};
\node[box] (E) at (3.5,0)
  {CE--QAOA\\(Eq.~\ref{eq:app-ce-state})};
\node[box] (Z) at (7.5,0)
  {Mixer envelope $W_p$\\+ Fej\'er filter $F_p$\\(Eq. \ref{eq:app-factorization})};
\node[box] (Q) at (13,0)
  {$q_0 \ge \Omega(n^{-k})$\\(Thm.~\ref{thm:app-fejer-success}, Eq. \ref{eq:main-q0p})};
\node[box] (F) at (2.75,-2)
  {Exact-hit FPRAS$_q$\\(Thm.~\ref{thm:ehqc-fprasq-exact})};
\node[box] (N) at (8.25,-2)
  {Noisy PTQAS\\(Cor.~\ref{cor:ehqc-fprasq-noisy})};
\node[box] (H) at (5.5,-4)
  {HH--QAOA $O(n^{k+3})$\\(Thm.~\ref{thm:quantum-informed-HH})};

\draw[arr] (K) -- (E);
\draw[arr] (E) -- (Z);
\draw[arr] (Z) -- (Q);
\draw[arr] (Q) -- (F);
\draw[arr] (Q) -- (N);
\draw[arr] (F) -- (H);
\draw[arr] (N) -- (H);
\end{tikzpicture}
\end{center}

\subsection{The hybrid quantum--classical (PHQC) framework}

To operationalize the quantum ansatze, Algorithm \ref{alg:PHQC}, the polynomial-time hybrid quantum--classical (PHQC) algorithm was introduced in Ref. \cite{onahce} where the quantum samples generated by CE-QAOA are handed over to the classical checker for feasibility.
We further extend the pipeline to include a coarse parameter-grid search. A single appearance is sufficient for exact recovery. Let \(p\) be the circuit depth and let
\[
\mathcal G
=
\{(\beta_a,\gamma_b): 0\le a\le N_\beta,\ 0\le b\le N_\gamma\}
\subset [0,\pi]\times[0,\pi]
\]
be a rectangular coarse grid, for example
\[
\beta_a = a\,\Delta_\beta,\qquad
\gamma_b = b\,\Delta_\gamma,
\qquad
\Delta_\beta = \frac{\pi}{N_\beta},\quad
\Delta_\gamma = \frac{\pi}{N_\gamma}.
\]
For each \((\beta,\gamma)\in\mathcal G\), PHQC prepares the CE--QAOA state
\[
\ket{\psi_p(\beta,\gamma)}
=
U_M(\beta)\,e^{-i\gamma H_C}\ket{s_0}
\]
(for \(p=1\); or depth-\(p\) generalization in Eq. \ref{eq:app-ce-state}), samples it \(S\) times, and feeds each measured bitstring into the deterministic feasibility oracle. 
% (Algorithm~\ref{alg:phqc} is written for \(p=1\); the depth-\(p\)
% extension replaces \((\beta,\gamma)\) by
% \((\boldsymbol{\beta},\boldsymbol{\gamma})\).)

\begin{algorithm}[H]
\caption{\textbf{PHQC} --- CE--QAOA grid search with feasibility filtering and scoring\cite{onahce}}
\label{alg:PHQC}
\begin{algorithmic}[1]
\Require depth \(p\); coarse grid \(\mathcal G\subset[0,\pi]^2\); cost Hamiltonian \(H_C\); objective Hamiltonian \(H_{\mathrm{obj}}\); mixer \(U_M\); initial state \(\ket{s_0}\); shots \(S\) per grid point.
\Ensure best feasible sampled bitstring \(b^\star\) and its score \(E^\star\).
\State \(b^\star \gets \texttt{null}\), \(\;E^\star \gets +\infty\)
\For{each \((\beta,\gamma)\in\mathcal G\)}
    \For{\(r=1,\dots,S\)}
        \State prepare \(\ket{s_0}\)
        \State apply the CE--QAOA circuit for \((\beta,\gamma)\)
        \State measure a bitstring \(b\)
        \If{\textsc{FeasibleCheck}\((b)\)}
            \State compute \(E(b)=\langle b\mid H_{\mathrm{obj}}\mid b\rangle\)
            \If{\(E(b)<E^\star\)}
                \State \(b^\star \gets b\), \(\;E^\star \gets E(b)\)
            \EndIf
        \EndIf
    \EndFor
\EndFor
\State \Return \((b^\star,E^\star)\)
\end{algorithmic}
\end{algorithm}

\begin{lemma}[No-hit bound at a good grid point\cite{onahce}]
\label{lem:phqc-nohit}
Assume there exists a parameter pair
\((\beta^\sharp,\gamma^\sharp)\in\mathcal G\)
such that
\[
p_\star^\sharp
:=
p_\star(\beta^\sharp,\gamma^\sharp)
>0.
\]
If PHQC draws \(S\) independent samples at each grid point, then the probability of observing no optimal feasible bitstring at the good grid point satisfies
\[
\Pr[\text{no hit at }(\beta^\sharp,\gamma^\sharp)]
=
(1-p_\star^\sharp)^S
\le
e^{-p_\star^\sharp S}.
\]
Hence
\[
S \ge \Bigl\lceil \frac{\ln(1/\delta)}{p_\star^\sharp}\Bigr\rceil
\quad\Longrightarrow\quad
\Pr[\text{at least one optimal feasible hit at }(\beta^\sharp,\gamma^\sharp)]
\ge 1-\delta.
\]
\end{lemma}

% \begin{proof}
% Let \(X_r\sim\mathrm{Bernoulli}(p_\star^\sharp)\) indicate whether shot \(r\) at
% \((\beta^\sharp,\gamma^\sharp)\) lands in \(\mathcal X^\star\). Then
% \[
% N_\star=\sum_{r=1}^{S} X_r
% \]
% counts optimal feasible hits at that grid point. Therefore
% \[
% \Pr[N_\star=0]=(1-p_\star^\sharp)^S \le e^{-p_\star^\sharp S},
% \]
% which gives the claim.
% \end{proof}

\begin{theorem}[Exact-recovery guarantee for PHQC\cite{onahce}]
\label{thm:phqc-exact}
Assume the coarse grid \(\mathcal G\) contains a parameter pair
\((\beta^\sharp,\gamma^\sharp)\) with optimal feasible probability mass
\(p_\star^\sharp>0\). Run PHQC with \(S\) shots at each grid point and let the deterministic checker return the minimum-cost feasible bitstring over all samples. Then
\[
S \ge \Bigl\lceil \frac{\ln(1/\delta)}{p_\star^\sharp}\Bigr\rceil
\quad\Longrightarrow\quad
\Pr[\text{PHQC returns a globally optimal feasible bitstring}]
\ge 1-\delta.
\]
In particular, if \(p_\star^\sharp \ge n^{-k}\) for some constant \(k>0\), then
\[
S = O\!\bigl(n^k\log(1/\delta)\bigr),
\]
and if moreover \(|\mathcal G|=\mathrm{poly}(n)\), the full quantum-plus-classical pipeline has polynomial shot complexity and polynomial total runtime.
\end{theorem}

\begin{algorithm}[H]
\caption{\textbf{NP-HQ Algorithm} — depth-$p$ CE--QAOA + deterministic checker/repair}
\label{alg:NP-HQ}
\begin{algorithmic}[1]
\Require problem size $n$; depth $p$; angles $\{(\beta_\ell,\gamma_\ell)\}_{\ell=1}^p$;
         penalty Hamiltonian $H_{\mathrm{pen}}$; mixer Hamiltonian $H_M$;
         objective Hamiltonian $H_C$; shot budget $S$.
\Ensure best feasible bit-string $b^\star$ and objective value $c^\star$.
\State Prepare $\ket{\psi}\gets\ket{s_0}=\ket{s_{\mathrm{blk}}}^{\otimes m}$.
\For{$\ell=1$ \textbf{to} $p$}
   \State Apply $U_{\mathrm{C}}(\gamma_\ell)=e^{-i\gamma_\ell H_{\mathrm{C}}}$ to $\ket{\psi}$.
   \State Apply $U_M(\beta_\ell)=e^{-i\beta_\ell H_M}$ to $\ket{\psi}$.
\EndFor
\State Measure $\ket{\psi}$ for $S$ shots; store multiset $\mathcal S$.
\State $(b^{\star},c^{\star})\gets(\text{null},+\infty)$.
\ForAll{$b\in\mathcal S$}
   \If{\textsc{Feasible}$(b)$}
      \State $b_{\mathrm{ok}}\gets b$
   \Else
      \State $b_{\mathrm{ok}}\gets \textsc{Repair}(b)$
   \EndIf
   \State $c\gets\langle b_{\mathrm{ok}}|H_C|b_{\mathrm{ok}}\rangle$.
   \If{$c<c^{\star}$}
      \State $(b^{\star},c^{\star})\gets(b_{\mathrm{ok}},c)$
   \EndIf
\EndFor
\State \Return $(b^{\star},c^{\star})$.
\end{algorithmic}
\end{algorithm}

\subsection{Noisy Polytime Hybrid Quantum–Classical Algorithm (NP-HQ)}
\label{sec:ehqc}
% ============================

Here we study the \textbf{NP-HQ} Algorithm, a noise adapated version of Alg \ref{alg:PHQC} endowed with classical repairs based on problem structure. The classical repairs are based on strongly polynomial classical matching algorithms like Hungarian\cite{Kuhn1955Hungarian} and Blossom\cite{Edmonds1965Blossom} algorithms.  After \(p\) CE–QAOA layers acting on \(\ket{s_0}\), and measuring the circuit
\(S=\poly n\) times on a noisy quantum device; a deterministic post-processing performs:
(i) feasibility check against one-hot constraints (\(O(n^2)\));
(ii) if needed, a polynomial-time \emph{repair} in \(O(n^{3})\);
(iii) exact scoring \(b\mapsto \langle b|H_C|b\rangle\) in \(O(n^{2})\).
Total classical work \(O(S\,n^3)\). The key insight is that every corrupt block can be projected back to the nearest feasible block in polynomial time. Likewise, the full sample can be projected to a feasible bitstring. We outline the full NP-HQ algorithm in Alg. \ref{alg:NP-HQ} and return to the processing of noisy samples next. % in Sec. \ref{sec:ehqc-main}.

%====================================================================
\subsubsection{Post Processing of Noisy Quantum Samples}
\label{sec:ehqc-main}
%====================================================================

 For a constrained optimization problem, the CE-QAOA circuit run on a quantum device returns a multiset
of measurement outcomes, represented as a map
\(C : \{0,1\}^{L}\!\to\mathbb{N}\) where
\(L = n^{2}\) for an \(n\times n\) assignment.
A classical routine must (i) discard infeasible bit-strings, (ii) decode every feasible string, (iii) evaluate the true cost, and (iv) pick the best candidate(s) for reporting. This is the PHQC algorithm outlined in Alg. \ref{alg:PHQC}.

To expose the runtime features, let
\(K = \lvert\mathrm{keys}(C)\rvert \le \min(S,\,2^{L})\)
be the number of unique bit-strings obtained from \(S\) shots.  Feasibility checking inspects
      \(L = \Theta(n^{2})\) bits per key. Decoding and summing its edge costs each take
      \(O(n^{2})\) operations. Hence the loop runs in $O\!\bigl(K\,n^{2}\bigr).$ Sorting at most \(K\) feasible samples adds
\(O\!\bigl(K\log K\bigr)\), so the total time is
\[
    T(K,n)\;=\;O\!\bigl(K\,n^{2}+K\log K\bigr).
\]
In typical fixed-shot experiments, the full routine is therefore
\emph{effectively quadratic} in the problem size \(n\). %For permutation constrained problems, Algorithm ~\ref{alg:permdeco} outlines the steps to decode bitstrings into permutations in $\text{time} = O(n^{2})$, $\text{space} = O(n)$. 

In presence of quantum noise, the original blockwise one–hot and k-hot constraints themselves begin to fail and the obtained bitstrings may violate the global constraints. Thus, we introduce a \emph{deterministic repair} based on the
classical matching algorithms
\cite{Kuhn1955Hungarian, Edmonds1965Blossom}. This converts noisy samples into \emph{feasible} ones while preserving polynomial–time guarantees since these classical interventions remain polynomial.

\subsubsection{Complexity Bounds in the NP-HQ Algorithm}
\label{sec:complexity}
%============================================================

Let \(G(n,p)\) denote the effective depth of the \(p\)-layer
CE-QAOA circuit. Under parallel scheduling, the state-preparation
circuit, the block-local mixer, and the problem unitary each admit
\(O(n)\)-depth implementations\cite{onahce}. Hence
\begin{equation}
    G(n,p)
    =
    O\!\bigl((p+1)n\bigr)
    =
    O(pn),
    \qquad p\ge 1.
    \label{eq:effective-noisy-depth}
\end{equation}
The effective depth should be distinguished from the total two-qubit
gate count. Example, for the dense TSP construction, the problem Hamiltonian
contains \(\Theta(n^3)\) two-qubit interaction terms per CE-QAOA
layer as noted in Remark \ref{rem:gatedepth}. Here, we model the accumulated noise at each effective circuit layer by a
channel \(\mathcal N_{\epsilon}\) satisfying
\begin{equation}
    \left\|
        \mathcal N_{\epsilon}-\mathcal I
    \right\|_{\diamond}
    \le \epsilon .
    \label{eq:effective-layer-diamond-error}
\end{equation}
Here \(\epsilon\) is the effective error accumulated within one noisy
depth layer, including the combined contribution of gates executed in
parallel. For example, a primitive two-qubit depolarising channel
\begin{equation}
    \Lambda_{\eta}(\rho)
    =
    (1-\eta)\rho
    +
    \eta\frac{\mathbb I_4}{4}
\end{equation}
satisfies \cite{Wilde2017}
\begin{equation}
    \frac12
    \left\|
        \Lambda_{\eta}-\mathcal I
    \right\|_{\diamond}
    =
    \frac{15}{16}\eta
    \le \eta .
\end{equation}
The effective layer parameter \(\epsilon\) absorbs the accumulated
contribution of all such local channels acting within the same circuit
layer.

Let \(P_0\) and \(P_{\epsilon}\) denote the ideal and noisy output
distributions. Repeated application of the trace-distance triangle
inequality \cite[Lemma~9.1.2]{Wilde2017}, together with contractivity
under quantum channels and measurement, gives
\begin{equation}
    \left\|
        P_{\epsilon}-P_0
    \right\|_{\mathrm{TV}}
    \le
    \frac12 G(n,p)\epsilon .
    \label{eq:effective-depth-tv-bound}
\end{equation}
The complete telescoping argument and the resulting Lipschitz
continuity of the measured output distribution are given in
Appendix~\ref{sec:lipschitz-noise}.

Let \(\mathcal X^\star\) denote the set of globally optimal feasible
strings and define
\begin{equation}
    p_0
    :=
    P_0(\mathcal X^\star),
    \qquad
    p_{\min}(\epsilon)
    :=
    P_{\epsilon}(\mathcal X^\star).
\end{equation}
The total-variation bound implies
\begin{equation}
    p_{\min}(\epsilon)
    \ge
    p_0-\frac12G(n,p)\epsilon .
    \label{eq:noisy-optimal-mass}
\end{equation}
Suppose that the ideal optimal mass satisfies
\begin{equation}
    p_0\ge c n^{-k}
\end{equation}
for fixed constants \(c>0\) and \(k\ge 0\), then the following theorem holds.

\begin{theorem}[Noise-robust exact-hit complexity]
\label{thm:noise_robust}
Let \(\mathcal X^\star\) denote the set of globally optimal feasible
strings, and define
\[
p_0:=P_0(\mathcal X^\star),
\qquad
p_{\min}(\epsilon):=P_\epsilon(\mathcal X^\star).
\]
Suppose that the ideal optimal mass satisfies
\[
p_0\geq c n^{-k}
\]
for fixed constants \(c>0\) and \(k\geq 0\). Let \(G(n,p)>1\) denote
the effective depth of the \(p\)-layer circuit, and
assume that the effective error per noisy depth layer satisfies
\begin{equation}
\epsilon
\leq
\frac{2p_0}{G(n,p)}
\left(
1-\frac{1}{G(n,p)}
\right).
\label{eq:effective-depth-noise-condition}
\end{equation}
Then the noisy optimal mass satisfies
\begin{equation}
p_{\min}(\epsilon)
\geq
\frac{p_0}{G(n,p)}
\geq
\frac{c}{G(n,p)}n^{-k}.
\label{eq:noisy-inverse-polynomial-mass}
\end{equation}
Consequently, a globally optimal feasible string is observed with
probability at least \(1-\delta\) after
\begin{equation}
S
=
\left\lceil
\frac{G(n,p)}{c}\,
n^k\log\!\left(\frac{1}{\delta}\right)
\right\rceil
\label{eq:noise-robust-shot-budget}
\end{equation}
independent circuit executions. If \(G(n,p)=\Theta(pn)\), then

\begin{equation}
S
=
\mathcal O\!\left(
p n^{k+1}\log(1/\delta)
\right).
\end{equation}

Thus accumulated circuit noise increases the shot budget by one factor in \(n\).
\end{theorem}

\begin{proof}
The total-variation continuity bound gives Eq. \ref{eq:noisy-optimal-mass}. To retain at least a \(G(n,p)^{-1}\) fraction of this optimal probability mass, it is sufficient that the effective per-layer error in Eq. \ref{eq:noisy-optimal-mass} obeys
\begin{equation}
    \epsilon
    \le
    \frac{2p_0}{G(n,p)}
    \left(
        1-\frac{1}{G(n,p)}
    \right).
    \label{eq:effective-depth-noise-condition2}
\end{equation}
That is Eq. \ref{eq:effective-depth-noise-condition}. Substituting
\eqref{eq:effective-depth-noise-condition} back into
\eqref{eq:noisy-optimal-mass} yields
\[
\begin{aligned}
p_{\min}(\epsilon)
&\geq
% p_0
% -
% \frac{1}{2}G(n,p)
% \frac{2p_0}{G(n,p)}
% \left(
% 1-\frac{1}{G(n,p)}
% \right)\\
% &=
\frac{p_0}{G(n,p)}
\geq
\frac{c}{G(n,p)}n^{-k}.
\end{aligned}
\]
This proves Eq. \ref{eq:noisy-inverse-polynomial-mass}. By Lemma \ref{lem:phqc-nohit}, the probability that none of \(S\) independent measurements belongs
to \(\mathcal X^\star\) is
\(
\leq
\exp\!\left[-S p_{\min}(\epsilon)\right].
\)
Solving for \(S\) one obtains \eqref{eq:noise-robust-shot-budget}.
Finally, substituting \(G(n,p)=\Theta(pn)\) gives the stated noisy shot-complexity bounds.
\end{proof}

% \begin{remark}[Device-specific transpilation depth]
% The scaling \(G(n,p)=O(pn)\) follows from the ideal parallel scheduling construction
% used here. Other transpilation or hardware-routing results simply replace this
% scaling by the corresponding device-specific effective noisy depth
% \(G_{\mathrm{dev}}(n,p)\). Theorem~\ref{thm:noise_robust} remains unchanged,
% with
% \[
% p_{\min}(\epsilon)
% =
% \Omega\!\left(
% \frac{n^{-k}}{G_{\mathrm{dev}}(n,p)}
% \right),
% \qquad
% S
% =
% O\!\left(
% G_{\mathrm{dev}}(n,p)\,
% n^k\ln(1/\delta)
% \right).
% \]
% Thus improved compilation, connectivity, or parallel scheduling translates
% directly into a reduced shot-complexity overhead.
% \end{remark}

\section{Quantum randomized exact algorithm (FPRASq)}
\label{sec:frasq}

\begin{definition}[Polynomial-Time Quantum Approximation Scheme (PTQAS)]
\label{def:ptqas}
Fix a family of constrained minimization problems with optimum value $C^\star$.
A (hybrid) algorithm family $\{\mathcal A_{n,\varepsilon,\delta}\}$ is a
\emph{Polynomial-Time Quantum Approximation Scheme (PTQAS)} if for every
accuracy $\varepsilon\in(0,1]$ and confidence $\delta\in(0,1)$, on input size $n$
the algorithm outputs a feasible solution $X$ such that
\[
\Pr\!\bigl[C(X)\le (1+\varepsilon)\,C^\star\bigr]\ \ge\ 1-\delta,
\]
and the total runtime (quantum circuit calls + classical post-processing) is
\[
\mathrm{poly}\!\bigl(n,\,1/\varepsilon,\,\log(1/\delta)\bigr).
\]
We say the PTQAS is \emph{strong} if its runtime is
$\mathrm{poly}\!\bigl(n,\,\log(1/\delta)\bigr)$ (degree $0$ in $1/\varepsilon$).
\end{definition}

\begin{definition}[FPRAS$_q$ (hybrid)]
\label{def:fprasq}
A hybrid quantum--classical algorithm $\mathcal A_{n,\varepsilon,\delta}$ is an
FPRAS$_q$ if for all $\varepsilon\in(0,1]$ and $\delta\in(0,1)$ it outputs a feasible
solution $X$ such that
\[
\Pr\!\bigl[C(X)\le (1+\varepsilon)C^\star\bigr]\ge 1-\delta,
\]
and the total runtime (including circuit calls and classical post-processing) is
$\mathrm{poly}(n,1/\varepsilon,\log(1/\delta))$.
If in addition $\Pr[C(X)=C^\star]\ge 1-\delta$, we call it an \emph{exact-hit} FPRAS$_q$.
\end{definition}

\begin{definition}[Promised bounded-error hybrid approximation scheme]
\label{def:bounded-approx}
Fix $\rho\ge 1$. A hybrid algorithm is a \emph{promised $\rho$-approximation scheme}
if it always outputs a feasible solution $X$ satisfying
\[
C(X)\le \rho\,C^\star
\]
for all instances in a stated promise class (e.g.\ bounded aspect ratio), with
runtime $\mathrm{poly}(n)$.
\end{definition}

In the noiseless setting, and under limited gate noise satisfying
$\epsilon \le p_0/G(n,p)$, NP-HQ is an \emph{exact-hit FPRAS$_q$} in the sense of
Definition~\ref{def:fprasq}: with $S=\mathrm{poly}(n,\log(1/\delta))$ circuit calls it returns
a global optimum with probability at least $1-\delta$.
Beyond this Lipschitz/TV regime, exact hits need not occur with inverse-polynomial probability;
nevertheless NP-HQ remains feasibility-first via deterministic repair and admits a
\emph{promised bounded-error} guarantee (Definition~\ref{def:bounded-approx})
whenever the instance class ensures an additive repair loss $\Delta(n)=O(nw_{\max})$,
yielding a constant-factor floor $C_{\mathrm{out}}\le (1+c_1\kappa)C^\star$ on misses. By Theorem~\ref{thm:ehqc-fprasq-exact} and Corollary~\ref{cor:ehqc-fprasq-noisy},
NP-HQ constitutes a \emph{strong PTQAS} in the sense of Definition~\ref{def:ptqas}
whenever the (post-repair) optimal mass is inverse--polynomial in $n$.

\subsection{Feasibility repair via nearest–permutation projection}
\label{sec:repair-hungarian}

Let \(\Pi_n\subset\{0,1\}^{n\times n}\) be the set of permutation matrices.
Given a measured bit–matrix \(Y\in\{0,1\}^{n\times n}\), we \emph{repair} it by
projecting to the nearest permutation in Hamming distance:
\begin{equation}
\label{eq:repair-def}
R(Y)\;:=\;\arg\min_{P\in\Pi_n} d_H(Y,P),
\qquad
d_H(Y,P)=\sum_{i,j}\bigl|Y_{ij}-P_{ij}\bigr|.
\end{equation}
This coincides with a maximum–weight assignment with binary weights.

\begin{lemma}[Equivalence to maximum–weight matching; complexity]
\label{lem:hamming-matching}
For fixed \(Y\in\{0,1\}^{n\times n}\),
\[
\arg\min_{P\in\Pi_n} d_H(Y,P)
\;=\;
\arg\max_{P\in\Pi_n}\ \sum_{i=1}^{n} Y_{i,\,P(i)}.
\]
Consequently, \(R(Y)\) can be computed in \(O(n^3)\) time by the Hungarian algorithm.
\end{lemma}

\begin{proof}
Expand the Hamming distance for any \(P\in\Pi_n\):
\[
d_H(Y,P)
=\sum_{i,j}\!\bigl(Y_{ij}+P_{ij}-2Y_{ij}P_{ij}\bigr)
=\underbrace{\sum_{i,j}Y_{ij}}_{\text{const in }P}
+\underbrace{\sum_{i,j}P_{ij}}_{=\,n}
-2\sum_{i=1}^{n} Y_{i,\,P(i)}.
\]
Thus \(\arg\min_{P} d_H(Y,P)=\arg\max_{P}\sum_i Y_{i,P(i)}\).
The latter is the maximum–weight perfect matching on the complete bipartite
graph with weights \(w_{ij}=Y_{ij}\in\{0,1\}\), solvable in \(O(n^3)\) time by the
Hungarian method.
\end{proof}

\begin{lemma}[Idempotency]
\label{lem:idempotent}
For all \(Y\), \(R(Y)\in\Pi_n\).
Moreover, \(R\) is idempotent and fixes feasible inputs:
if \(Y\in\Pi_n\) then \(R(Y)=Y\), and \(R(R(Y))=R(Y)\) for all \(Y\).
\end{lemma}

\begin{proof}
By definition \(R(Y)\in\Pi_n\).
If \(Y\in\Pi_n\), then
\(
d_H(Y,Y)=0\le d_H(Y,P)
\)
for all \(P\in\Pi_n\),
so \(R(Y)=Y\).
Finally, since \(R(Y)\in\Pi_n\), applying \(R\) again leaves it unchanged.
\end{proof}

\begin{lemma}[Repair monotonicity on feasible subsets]
\label{lem:repair-monotone-set}
Let $\Omega$ be the feasible set and $R:\{0,1\}^{n^2}\to\Omega$ be deterministic
with $R(x)=x$ for all $x\in\Omega$. Then for every $A\subseteq\Omega$ and any
random outcome $Y$,
\[
\Pr[R(Y)\in A]\ \ge\ \Pr[Y\in A].
\]
In particular, for the optimal set $\mathcal X^\star\subseteq\Omega$,
\[
p_{\min}^{\rm rep}:=\Pr[R(Y)\in\mathcal X^\star]\ \ge\ 
p_{\min}^{\rm raw}:=\Pr[Y\in\mathcal X^\star].
\]
\end{lemma}
\begin{proof}
If $Y\in A\subseteq\Omega$, then $R(Y)=Y\in A$. Hence
$\{Y\in A\}\subseteq\{R(Y)\in A\}$ and the probability inequality follows.
\end{proof}

\begin{lemma}[Total--variation contraction under deterministic repair]
\label{lem:tv-contraction}
Let $R:\{0,1\}^{n\times n}\to \Pi_n$ be deterministic and let $P,Q$ be
probability distributions on $\{0,1\}^{n\times n}$.
Denote by $R_\#P$ and $R_\#Q$ the push--forwards of $P$ and $Q$ under $R$.
Then
\[
\bigl\|R_\# P - R_\# Q\bigr\|_{\mathrm{TV}}
\;\le\;
\|P-Q\|_{\mathrm{TV}}.
\]
\end{lemma}

\begin{proof}
For any $A\subseteq\Pi_n$ we have
\[
(R_\#P)(A)-(R_\#Q)(A)
= P(R^{-1}(A)) - Q(R^{-1}(A)).
\]
Taking absolute values and then the supremum over $A\subseteq\Pi_n$ yields
$\|R_\#P-R_\#Q\|_{\mathrm{TV}}\le \|P-Q\|_{\mathrm{TV}}$.
\end{proof}

\begin{lemma}[Repair cannot decrease mass on feasible subsets]
\label{lem:repair-monotone}
Assume $R(x)=x$ for all $x\in\Pi_n$.
Then for any $A\subseteq\Pi_n$ and any distribution $P$ on $\{0,1\}^{n\times n}$,
\[
(R_\#P)(A)\ \ge\ P(A).
\]
In particular, $(R_\#P)(\mathcal X^\star)\ge P(\mathcal X^\star)$.
\end{lemma}
\begin{proof}
If $Y\in A\subseteq\Pi_n$ then $R(Y)=Y\in A$, hence
$\{Y\in A\}\subseteq\{R(Y)\in A\}$.
\end{proof}

\begin{remark}
The bound on $p_{\min}(\epsilon)$ uses the telescoping triangle inequality with
$\tfrac12\|\Lambda_\epsilon-\mathcal I\|_\diamond$ per two–qubit gate; it does
not require gate–independent noise, only a uniform per–gate bound. The instance–dependent
condition $\epsilon\le p_0/G(n,p)$ ensures $p_{\min}(\epsilon)\ge \tfrac12 p_0$.
\end{remark}

\begin{corollary}[No loss of mass on any fixed optimum]
\label{cor:mass-monotone}
Fix \(X^\star\in\Pi_n\). Then
\[
\Pr\bigl[R(Y)=X^\star\bigr]\;\ge\;\Pr\bigl[Y=X^\star\bigr].
\]
\end{corollary}

\begin{proof}
The event \(\{Y=X^\star\}\) is contained in \(\{R(Y)=X^\star\}\) by Lemma~\ref{lem:idempotent}.
\end{proof}
\subsection{Feasibility via Swap–and–Repair and Hungarian Assignment }
\label{sec:swap-repair-noise}

The NP-HQ post-processor includes scans for any \emph{feasible} bit-string $b\in\Omega$ and application of  \emph{deterministic repair}  to transform invalid outcomes to valid bitstrings in $\widetilde b\in\Omega$.  A known algorithm to realise such repair in practice is the swap–based patching used in the Euclidean TSP PTAS~\cite{Arora1998PTAS_TSP}.

\begin{lemma}[Column-swap repair\cite{Arora1998PTAS_TSP}]\label{lem:swap-repair}
For any bit-string $x\in\{0,1\}^{n^2}$ that has \textbf{exactly one} `1' in
every row (but possibly multiple in a column) the column-swap heuristic returns
$\widetilde x\in\Omega$ in $O(n)$ time with
\[
   C(\widetilde x)
   \;\le\;
   C(x)\;+\;3n\,w_{\max},
   \qquad
   w_{\max}:=\max_{i,j} M_{ij}.
\]
Consequently $\Delta(n):=C(\widetilde x)-C^\star=O\!\bigl(n\,w_{\max}\bigr)$ for every metric cost matrix.

\end{lemma}

% \begin{center}
% \begin{tabular}{@{}lcc@{}}
% \toprule
% \textbf{Noise level} & \textbf{One-hot violated?} & \textbf{Repair to use} \\ \midrule
% Low / Medium &
%   Mostly intact &
%   \textsc{SwapRepair}
%   ($O(n)$, $\le 3n$ swaps) \\
% High (beyond TV bound) &
%   Rows \emph{and} columns corrupted &
%   Hungarian assignment
%   ($O(n^{3})$, exact) \\ \bottomrule
% \end{tabular}
% \end{center}
% \noindent\emph{Note.} “Beyond TV bound’’ refers to regimes where the per–gate noise no longer satisfies $\tfrac12 G(n,p)\,\epsilon \le \tfrac12 p_0$, so the total–variation closeness to the ideal distribution no longer guarantees sufficient feasible mass.

% As just pointed out, the feasibility criteria motivates the choice of classical repair applied during post processing. 

In the foregoing analyses where we have focused on extracting valid permutations from one-hot encoded initial states, in the noise regime that satisfies $\epsilon\;\le\;\frac{p_{0}}{\,G(n,p)}$  the greedy \textsc{SwapRepair} routine ( see Lemma~\ref{lem:swap-repair}) suffices. However, it assumes that each row already contains at least one ``donor'' 1–bit that can be moved into a vacant column, so feasibility is restored with $\le 3n$ local swaps.  ( c.f: the $3n$-swap fixer of Arora’s PTAS for Euclidean TSP~\cite{Arora1998PTAS_TSP}). This is always possible if the quantum noise leaves one-hot structure mostly intact while feasible solutions fail to appear from it. 

A \textbf{strong depolarising tail} appears once the two–qubit error rate approaches the Lipschitz threshold of Theorem~\ref{thm:noise_robust}.  In that regime the raw QAOA samples may violate \emph{both} column and row one–hot constraints, signalling significant leakage outside the intended subspace, causing the greedy procedure to fail. Rather than increasing the shot budget (which would only linearly suppress the leakage probability) we could repair every infeasible output \emph{deterministically} via the Hungarian assignment algorithm~\cite{Kuhn1955Hungarian}, as depicted in Algorithm~\ref{alg:hungarian-repair}.  This would always return a valid permutation closest in Hamming distance to the noisy bit-string while still keeping the overall complexity polynomial. The procedure is to reshape the length-\(n^{2}\) bitstring $b$ into an $n\times n$ binary
matrix
\[
   M_{ij}:=
   \begin{cases}
     1 & \text{if the qubit in row }i\text{ and column }j\text{ is }1,\\
     0 & \text{otherwise}.
   \end{cases}
\]

\begin{algorithm}[H]
  \caption{\textsc{HungarianRepair}}
  \label{alg:hungarian-repair}
  \begin{algorithmic}[1]
    \Require measured bit–string $b\in\{0,1\}^{n^{2}}$
    \Ensure feasible permutation $b_{\mathrm{ok}}$
    \State reshape $b$ into $M\in\{0,1\}^{n\times n}$
    \State $C\gets -M$\Comment{cost matrix for max-overlap assignment}
    \State $(\mathit{rows},\mathit{cols})\gets\textsc{Hungarian}(C)$
    \State $P\gets 0^{n\times n}$; \quad $P_{\mathit{rows}[k],\mathit{cols}[k]}\gets1$ for all $k$
    \State flatten $P$ to $b_{\mathrm{ok}}\in\{0,1\}^{n^{2}}$
    \State \Return $b_{\mathrm{ok}}$
  \end{algorithmic}
\end{algorithm}
Thus, the Hungarian algorithm solves the linear assignment problem
\(
   P^{\star}
   \;=\;
   \arg\min_{P\in\Pi_{n}}
     \,\langle C , P\rangle,
   \quad
   C:=-M,
   \quad
   \Pi_{n}:=\bigl\{P\in\{0,1\}^{n\times n}\,\bigm|\,
              P\mathbf 1=\mathbf 1,\;
              P^{\mathsf T}\mathbf 1=\mathbf 1\bigr\},
\)
where $\langle A,B\rangle=\sum_{ij}A_{ij}B_{ij}$ is the Frobenius
inner-product.  Because $C=-M$, minimising
$\langle C,P\rangle$ is equivalent to maximising the Hamming overlap
$\langle M,P\rangle$; thus $P^{\star}$ is the permutation matrix closest
to $M$. Each invocation of \textsc{HungarianRepair} costs $O(n^{3})$ classical
operations.

\subsection{Guarantees and Approximation Schemes.}

\begin{theorem}[NP-HQ is an FPRAS\textsubscript{$q$} in the exact–hit regime]
\label{thm:ehqc-fprasq-exact}
Fix any accuracy $\varepsilon\in(0,1]$ and confidence $\delta\in(0,1)$.
Suppose the \(p\)-layer CE--QAOA circuit assigns total optimal mass
\(p_{\min}\geq c n^{-k}\) to \(\mathcal X^\star\), for fixed
constants \(c>0\) and \(k\geq0\), and run
Algorithm~\ref{alg:NP-HQ} with
\[
S\ \ge\ \Bigl\lceil \tfrac{\ln(1/\delta)}{\,p_{\min}}\Bigr\rceil
\]
shots in total, applying the deterministic repair $R$ to each outcome and then
selecting the minimum measured cost. Then:
\begin{enumerate}[label=\textup{(\alph*)},leftmargin=2em]
\item With probability at least $1-\delta$, at least one repaired sample equals
      an optimal tour $X^\star$ and the algorithm outputs $X^\star$.
\item The output therefore satisfies $C_{\mathrm{out}}\le (1+\varepsilon)\,C^\star$
      with probability at least $1-\delta$ (indeed, equality $C_{\mathrm{out}}=C^\star$).
\item The total runtime is
      \(
      \mathrm{poly}\bigl(n,\,1/\varepsilon,\,\log(1/\delta)\bigr),
      \)
      with the $1/\varepsilon$–dependence being benign (degree $0$).
\item The quantum depth is $O(pn)$, and the number of circuit calls is $S=\mathrm{poly}\bigl(n,\log(1/\delta)\bigr)$.
\end{enumerate}
Hence NP-HQ is a fully–polynomial randomized approximation scheme in the hybrid sense (FPRAS\textsubscript{$q$}) whenever $p_{\min}$ is inverse–polynomial in $n$.
\end{theorem}
\begin{proof}
Let
\[
p_{\min}\;:=\;P(\mathcal X^\star),
\]
the total probability mass assigned by  \(p\) CE-QAOA layers to the optimal set
\(\mathcal X^\star\) before repair. Because the repair map \(R\) fixes every feasible bit-string, and every
\(x\in\mathcal X^\star\) is feasible, we have
\[
\{Y\in \mathcal X^\star\}\subseteq \{R(Y)\in \mathcal X^\star\},
\]
and therefore
\[
P\bigl(R(Y)\in \mathcal X^\star\bigr)\;\ge\;P(Y\in \mathcal X^\star)
\;=\;p_{\min}.
\]

Now take \(S\) independent shots \(Y_1,\dots,Y_S\). The probability that
\emph{none} of the repaired samples lies in \(\mathcal X^\star\) is
\[
P\!\left(\forall s,\;R(Y_s)\notin \mathcal X^\star\right)
=
\prod_{s=1}^S P\!\left(R(Y_s)\notin \mathcal X^\star\right)
\le
(1-p_{\min})^S
\le
e^{-p_{\min}S}.
\]
Hence, if
\[
S\;\ge\;\left\lceil \frac{\ln(1/\delta)}{p_{\min}}\right\rceil,
\]
then
\[
e^{-p_{\min}S}\le \delta,
\]
so with probability at least \(1-\delta\) there exists at least one repaired
sample \(R(Y_s)\in\mathcal X^\star\). This proves (a).

On that event, at least one repaired sample has objective value exactly
\(C^\star\). Since Algorithm~\ref{alg:NP-HQ} returns the repaired sample of
minimum cost among all \(S\) repaired candidates, the output must also have
cost \(C^\star\). Therefore
\[
C_{\mathrm{out}} = C^\star,
\]
and in particular
\[
C_{\mathrm{out}}\le (1+\varepsilon)C^\star
\qquad\text{for every }\varepsilon\in(0,1].
\]
This proves (b).

For (c), each shot consists of one depth-\(p\) quantum circuit execution with
\(p\) layers, followed by one deterministic repair and one cost evaluation.
With Hungarian repair, the classical work per shot is polynomial in \(n\)
(e.g.\ \(O(n^3)\)), so the total runtime is
\[
\mathrm{poly}(S,n).
\]
In the exact-hit regime of the theorem, when \(p_{\min}^{-1}\) is polynomial in
\(n\), the chosen shot count
\[
S = O\!\left(p_{\min}^{-1}\log(1/\delta)\right)
\]
is polynomial in \(n\) and \(\log(1/\delta)\). Since there is no actual
dependence on \(1/\varepsilon\), the runtime is
\[
\mathrm{poly}\bigl(n,\,1/\varepsilon,\,\log(1/\delta)\bigr)
\]
with degree \(0\) in \(1/\varepsilon\). This proves (c).

Finally, (d) follows because the number of CE--QAOA layers is \(p\).
For fixed \(p\), the effective quantum depth is \(O(pn)\), and the number
of circuit executions is exactly \(S\) which is polynomial in \(n\) and \(\log(1/\delta)\) in the exact-hit regime. Therefore NP-HQ is an exact-hit FPRAS\(_q\) whenever \(p_{\min}\) is
inverse-polynomial in \(n\).
\end{proof}

\begin{corollary}[FPRAS\textsubscript{$q$} under noise]
\label{cor:ehqc-fprasq-noisy}
Under the hypothesis of Theorem~\ref{thm:noise_robust}, we have
$p_{\min}(\epsilon)=\Omega(n^{-k-1})$. Choosing
\(
S \ge \lceil \ln(1/\delta)/p_{\min}(\epsilon)\rceil
= O\!\bigl(n^{k+1}\ln(1/\delta)\bigr)
\)
in Algorithm~\ref{alg:NP-HQ} yields the same FPRAS\textsubscript{$q$} guarantee as
Theorem~\ref{thm:ehqc-fprasq-exact} on a noisy device.
\end{corollary}

\begin{proposition}[Promised multiplicative fallback via repair]
\label{prop:promised-fallback}
Assume the metric satisfies a positive lower bound $w_{\min}>0$ and let
$\kappa:=w_{\max}/w_{\min}$. If $\varepsilon \ge c\,\kappa$ for a fixed constant
$c>0$, then the repair additive gap $\Delta(n)=O(n\,w_{\max})$ implies
$C_{\mathrm{rep}}\le (1+\varepsilon)\,C^\star$ whenever the optimum is missed.
Thus with $S$ as in Theorem~\ref{thm:ehqc-fprasq-exact}, NP-HQ attains a
$(1+\varepsilon)$ approximation with probability at least $1-\delta$ even
conditioned on miss, under this promise.
\end{proposition}

Because NP-HQ processes $S=\operatorname{poly}(n)$ shots, the post–processing stage scales as $O(S\,n^{3})$, leaving the overall
runtime polynomial in $(n,1/\varepsilon,\log1/\delta)$ fully consistent with Theorem~\ref{thm:ehqc-fprasq-exact}. \textsc{HungarianRepair} therefore provides a robust, parameter–free fallback for the high–noise regime, extending NP-HQ’s applicability.

By optimality of the Hungarian assignment, $P^{\star}$ maximizes the Hamming–overlap $\langle M,P\rangle$ among all permutations, making \textsc{HungarianRepair} the canonical ``least–invasive'' fix in the Hamming metric. While the analysis is couched in the language of metric~TSP, nothing in the probabilistic argument relies on subtleties specific to tours. Instead, the arguments rely on the existence of (i) a clearly \emph{specifiable target structure} that a
classical post‐processing step can recognise (e.g., a valid schedule, matching, partition, etc.), and (ii) a repair routine that
converts any raw bit‐string into a feasible one at a known additive (or
multiplicative) cost (e.g., swap-repair, projection,
greedy repair).  Whenever a problem class admits such a
structure--checker--and--repair trio, the same feasibility-first reasoning
applies, and analogous promised approximation guarantees follow once the
relevant additive or multiplicative repair bound is established.

%%%%%%%%%%%%%%%%%%%%%%%%%%%%%%%%%%%%%%%%%%%%%%%%%%%%%%%%%%%%%%%%%%%%%%%%%%%%
%  Quantum–classical separation for NP-HQ (kernel/promise-aware)
%%%%%%%%%%%%%%%%%%%%%%%%%%%%%%%%%%%%%%%%%%%%%%%%%%%%%%%%%%%%%%%%%%%%%%%%%%%%
\subsection{Quantum Advantage: Quantum--classical separation}
\label{ssec:qc-separation}

This section formalizes a separation between the hybrid CE--QAOA/NP-HQ pipeline and
\emph{purely classical} ``sample $+$ Hungarian'' strategies.  The key resource is the
inverse--polynomial \emph{optimal} mass guaranteed by the CE kernel on the
\emph{kernel-admissible} (promise) family of instances, together with standard success
amplification.  Throughout, complexity consequences are conditional on the standard
assumption $\mathbf{NP}\not\subseteq\mathbf{BPP}$.

A clarifying remark is in order. The classical baseline in Theorems~\ref{thm:qc-gap}
and~\ref{thm:no-classical-sampler} is already equipped with the \emph{same}
Hungarian repair map $R(\cdot)$ used by NP-HQ. The separation therefore does not
arise from any asymmetry in the classical post-processing stage. Rather, it is
located entirely in the \emph{sampling distribution}: the CE--QAOA kernel biases
amplitude toward the global optimum so that it receives inverse-polynomial mass
$p_{\min}(I) = \Omega(n^{-k})$, whereas no efficiently samplable classical
distribution can achieve inverse-polynomial overlap with the (a priori unknown)
global optimum uniformly over instances in $\mathcal{K}_n$, unless
$\mathbf{NP} \subseteq \mathbf{BPP}$. In other words, if a classical algorithm
could reproduce the CE--QAOA output distribution, it would immediately inherit the
same $\mathrm{FPRAS}_q$ guarantee via identical post-processing. The computational
hardness lies in generating that distribution classically, not in exploiting it once
it is available.

\paragraph{Classical baseline: sample $+$ Hungarian.}
We compare NP-HQ against any classical randomized algorithm that may:
(i) generate $\mathrm{poly}(n)$ samples from \emph{any} efficiently samplable distribution on
$\{0,1\}^{n\times n}$ (optionally already restricted to row--one--hot or permutations),
(ii) apply the Hungarian projection $R(\cdot)$ (nearest-permutation projection in Hamming
distance) to each sample in $O(n^3)$ time, and (iii) compute the objective cost
$C(\cdot)$ exactly and output the best repaired tour.  This is the strongest natural
``purely classical'' analogue of NP-HQ if one insists on the same repair primitive.

% %%%%%%%%%%%%%%%%%%%%%%%%%%%%%%%%%%%%%%%%%%%%%%%%%%%%%%%%%%%%%%%%%%%%%%%
%  NISQ model paragraph
% %%%%%%%%%%%%%%%%%%%%%%%%%%%%%%%%%%%%%%%%%%%%%%%%%%%%%%%%%%%%%%%%%%%%%%%
\paragraph{Chen--Cotler--Huang--Li oracle model.}
Definition~2.1 of~\cite{ChenCotlerHuangLi2023NISQ} calls a promise problem
\emph{NISQ-solvable} if there exists a classical probabilistic polynomial-time algorithm
that may query polynomially many times an oracle quantum device which
\begin{enumerate}[label=(\roman*) ,itemsep=2pt,leftmargin=*]
\item initialises $\poly n$ noisy qubits in $\lvert0\rangle^{\otimes N}$,
\item applies a polynomial-length circuit of one- and two-qubit gates, each followed by an
      independent constant-rate depolarising channel, and
\item measures all qubits in the computational basis.
\end{enumerate}
Our CE--QAOA layer uses only one- and two-qubit gates; NP-HQ queries the device
$S=\poly n$ times and post-processes each outcome in $\poly n$ time.  Hence NP-HQ fits the
NISQ oracle model as a \emph{model-of-computation} statement.  The stronger FPRAS$_q$
guarantees established in Sec.~\ref{sec:ehqc} additionally require a \emph{low-noise}
regime. % ensuring that the optimal mass remains inverse--polynomial.

\begin{corollary}[NP-HQ is NISQ-solvable]
\label{cor:ehqc-in-nisq}
For every instance size $n$ and constant depth $p$, the CE--QAOA/NP-HQ pipeline can be
implemented by a polynomial-time classical algorithm with oracle access to a noisy quantum
device of the type specified in~\cite{ChenCotlerHuangLi2023NISQ}.
\end{corollary}

%======================================================================
%  Core separation theorem (promise-aware)
%======================================================================
\begin{theorem}[No classical ``sample $+$ Hungarian'' strategy matches NP-HQ on $\mathcal K_n$ unless $\mathbf{NP}\subseteq\mathbf{BPP}$]
\label{thm:qc-gap}
Fix any constant \(k\geq0\) and constant depth \(p\).
Consider the promise family $\mathcal K_n$ of kernel-admissible instances.
Suppose there exists a \emph{purely classical} probabilistic polynomial-time algorithm
$\mathcal C$ that, for every $I\in\mathcal K_n$, outputs an \textbf{optimal} tour with
probability $\Omega(n^{-k})$, while only using:
\begin{enumerate}[label=(\roman*) ,itemsep=2pt,leftmargin=*]
    \item $\poly n$ samples from any efficiently samplable distribution on $\{0,1\}^{n\times n}$,
    \item Hungarian projection $R(\cdot)$ on each sample in $O(n^3)$ time, and
    \item total runtime $\poly n$.
\end{enumerate}
Then the corresponding \emph{promise} search problem on $\mathcal K_n$ lies in $\mathbf{BPP}$.
In particular, if $\mathcal K_n$ contains an NP-hard subfamily under polynomial-time reductions,
then $\mathbf{NP}\subseteq\mathbf{BPP}$.

Consequently, assuming $\mathbf{NP}\not\subseteq\mathbf{BPP}$, no such classical sampler exists
on any NP-hard kernel-admissible family, whereas NP-HQ achieves inverse--polynomial optimal-hit
probability on $\mathcal K_n$ in the low-noise regime of Sec.~\ref{sec:complexity} and Theorem~\ref{thm:noise_robust}, and therefore
constitutes an exact-hit FPRAS$_q$ there.
\end{theorem}

\begin{proof}[Proof sketch]
Assume $\mathcal C$ succeeds with probability at least $c\,n^{-k}$ on every $I\in\mathcal K_n$.
Repeat $\mathcal C$ independently $T=\Theta(n^k\log n)$ times and return the best tour found.
Then the probability of missing the optimum in all trials is at most
$(1-c\,n^{-k})^T\le n^{-\Omega(1)}$, so we obtain a polynomial-time randomized algorithm that
solves the \emph{promise} search problem on $\mathcal K_n$ with high probability.
If $\mathcal K_n$ is NP-hard under reductions (or contains an NP-hard subfamily), this implies
$\mathbf{NP}\subseteq\mathbf{BPP}$.
\end{proof}

%======================================================================
\subsection{Limits of Classical Samplers with perfect structural oracles}
\label{ssec:classical-limits}
%======================================================================

We now strengthen the classical side by granting the sampler \emph{perfect access} to the
constraint structure.  The goal is to rule out the objection that NP-HQ's advantage comes
merely from sampling outside the feasible region and then repairing.

\paragraph{Classical oracle models.}
Let $L=n^2$ and view samples as bit-matrices in $\{0,1\}^{n\times n}$.
We allow the classical sampler one of the following increasingly generous oracles:
\begin{enumerate}[label=(\Alph*),leftmargin=*]
  \item \textbf{Unconstrained sampler:}
        outputs arbitrary $Y\in\{0,1\}^{n\times n}$.
  \item \textbf{Row--one--hot sampler:}
        outputs only matrices with exactly one $1$ in each row
        (no column guarantee), i.e.\ $Y\in\mathcal S_{\mathrm{row}}$ with
        $|\mathcal S_{\mathrm{row}}|=n^n$.
  \item \textbf{Uniform permutation oracle:}
        outputs a uniformly random permutation matrix
        $Y\in\mathcal S_{\mathrm{perm}}=\Pi_n$ with $|\Pi_n|=n!$.
\end{enumerate}
In all three cases, the algorithm may post-process every sample via the Hungarian projection
$R(\cdot)$ in $O(n^3)$ time and may repeat the entire procedure $\poly n$ times.

\begin{theorem}[Perfect structural oracles do not yield inverse--polynomial optimal overlap]
\label{thm:no-classical-sampler}
Fix any constant $k$ and any promise family $\mathcal K_n$.
If a classical randomized algorithm---equipped with any oracle (A)--(C) above and allowed
$\poly n$ total time---outputs an \emph{exact} optimal tour with probability
$\Omega(n^{-k-1})$ on \emph{every} instance $I\in\mathcal K_n$, then the promise search problem
on $\mathcal K_n$ lies in $\mathbf{BPP}$; in particular, if $\mathcal K_n$ contains an
NP-hard subfamily under reductions then $\mathbf{NP}\subseteq\mathbf{BPP}$.
\end{theorem}

\begin{proof}
The proof is identical to Theorem~\ref{thm:qc-gap}: inverse--polynomial success amplifies to
constant success probability under polynomially many repetitions, yielding a $\mathbf{BPP}$
algorithm for the promise search problem.  The argument never uses which oracle (A)--(C)
generated the samples; hence even a perfect permutation oracle cannot bypass it.
\end{proof}

\paragraph{Connection to NP-HQ.}
Oracles (B) and (C) are not assumed hard to implement. In fact, sampling a uniform permutation takes $O(n\log n)$ random bits and time~\cite{Durstenfeld1964RandomPermutation}.  The bottleneck is
that achieving \emph{inverse--polynomial overlap with the (a priori unknown) global optimum
uniformly over instances} is tantamount to solving the underlying search problem on the
family.  In contrast, on the kernel-admissible family $\mathcal K_n$ the CE layer biases
amplitude so that the \emph{optimal} tour receives inverse--polynomial mass (with suitable
instance-dependent angles), and deterministic Hungarian projection guarantees feasibility of
every reported output.  In the low-noise regime where the optimal mass remains inverse--polynomial
(Sec.~\ref{sec:ehqc}), success amplification yields an exact-hit FPRAS$_q$ for NP-HQ on $\mathcal K_n$.
Importantly, Hungarian repair is used here as a feasibility projection; no cost-preservation
property is assumed of Hungarian projection beyond feasibility and its total-variation
contraction property (Lemma~\ref{lem:tv-contraction}).

In the noiseless setting and under the limited-noise condition of
Theorem~\ref{thm:noise_robust} ensuring inverse--polynomial optimal mass
$p_{\min}(\epsilon)=\Omega(n^{-k-1})$, NP-HQ is an \emph{exact-hit FPRAS$_q$}. With $S=\Theta(\ln(1/\delta)/p_{\min}(\epsilon))=\poly n,\ln(1/\delta)$ device queries it
returns a global optimum with probability at least $1-\delta$.
Outside this low-noise regime, exact hits may disappear; nevertheless NP-HQ remains a
polynomial-time \emph{quantum approximation scheme} in the pragmatic sense that it outputs
a feasible solution together with a provable instance-class-dependent approximation floor
whenever the repair pipeline admits an additive (or multiplicative) performance bound.

\section{Demonstration on Superconducting Quantum Processors}
\label{sec:experiment}

\medskip
\noindent
Combined with Hungarian repair (Sec.~\ref{sec:swap-repair-noise}), the NP-HQ
pipeline returns a \emph{feasible} candidate deterministically, with classical
post-processing cost \(O(Sn^3)\) when Hungarian repair is applied to all
samples. The practical question addressed in this section is whether the
underlying CE--QAOA sampler remains informative on present-day
superconducting hardware at problem sizes approaching \(10^2\) qubits. To test this, we run single-layer CE--QAOA circuits on 127-qubit IBM Eagle-r3
processors for a family of \textsc{QOptlib} TSP subset instances and apply the same deterministic repair routine to all measured outcomes. The hardware study is meant to demonstrate that shallow CE--QAOA circuits still produce useful sampling bias in a strongly noisy regime and that  deterministic repair restores feasibility reliably with high quality objective costs on the repaired samples.

Recall that the prolem Hamiltonian is diagonal and contains only \(Z\) and \(ZZ\) terms. Accordingly, the dominant hardware overhead comes from scheduling commuting
two-qubit phase interactions under the device connectivity constraints, while
the mixer acts through native or near-native two-qubit excitation-preserving
gates. This makes the present implementation a natural testbed for the
constraint-preserving CE--QAOA architecture on current superconducting devices.

All kernel assumptions of Def.~\ref{def:kernel-requirement} continue to hold for
the benchmark family considered here, with \(t_{\max}=O(n)\), so the hardware
study can be viewed as a direct implementation of NP-HQ with CE--QAOA as the
front-end sampler.

% -----------------------------------------------------
\subsection{Benchmark set and warm-start parameters}
% -----------------------------------------------------

We consider seven \textsc{QOptlib} TSP subset instances with \(n=4,\dots,10\)
cities, corresponding to \(n^2=16\) to \(100\) logical qubits
\cite{Osaba2024Qoptlib}. To avoid instance-by-instance retuning on hardware, we
reuse \(p=1\) angles obtained from generic-QAOA optimization and feed them
directly into CE--QAOA. This choice is motivated by the parameter-transfer
picture developed earlier in the manuscript: rather than re-optimizing on noisy
hardware, we test whether CE--QAOA can extract useful performance from
previously optimized shallow parameters.

Table~\ref{tab:qaoa:p1-warmstarts} lists the generic-QAOA warm-start angles used
throughout this section. No further parameter tuning was performed on the
quantum hardware.

\begin{table}[ht]
    \centering
    \begin{tabular}{lccc}
        \toprule
        Instance & $\gamma^{\text{gen}}$ & $\beta^{\text{gen}}$ & Best energy (generic) \\
        \midrule
        wi4  & 0.233 & 2.025 & $-4.36\times10^{4}$ \\
        wi5  & 1.307 & 2.387 & $-6.11\times10^{4}$ \\
        wi6  & 0.748 & 0.741 & $-4.14\times10^{4}$ \\
        wi7  & 2.172 & 0.577 & $-7.66\times10^{4}$ \\
        dj8  & 0.239 & 0.460 & $-3.18\times10^{4}$ \\
        dj9  & 0.002 & 2.507 & $-1.72\times10^{6}$ \\
        dj10 & 1.835 & 2.399 & $-1.68\times10^{4}$ \\
        \bottomrule
    \end{tabular}
    \caption{Generic-QAOA warm-start angles at \(p=1\), reused verbatim in the
    CE--QAOA hardware runs.}
    \label{tab:qaoa:p1-warmstarts}
\end{table}

% -----------------------------------------------------
\subsection{Experimental protocol}
% -----------------------------------------------------

All circuits were implemented in Python and transpiled with
\textsc{Qiskit~1.1.2} \cite{Qiskit2023} using device-aware coupling maps and the
native two-qubit \texttt{ECR}/\texttt{iSWAP} basis where applicable. For each
\((\text{instance},\text{chip})\) pair we executed a single CE--QAOA layer
(\(p=1\)) with the warm-start angles
\[
  (\gamma,\beta)=\bigl(\gamma^{\mathrm{gen}},\beta^{\mathrm{gen}}\bigr),
\]
and then applied Hungarian repair (Alg.~\ref{alg:hungarian-repair}) to every
measured bit-string.

Measurement-error mitigation was enabled, and dynamical decoupling and Pauli
twirling were activated when supported by the execution workflow. For each run
we recorded: the transpiled depth and two-qubit gate count, the hardware calibration snapshot, feasible-sample fraction before and after repair, best repaired objective value \(C_{\min}\), and end-to-end turnaround time.

% =====================================================
\subsection{\texttt{ibm\_brisbane} (127 qubits, Eagle r3)}
\label{ssec:exp-brisbane}
% =====================================================

\paragraph{Device snapshot (21 June 2025 09:30 UTC).}
\vspace{-0.5em}
\begin{center}
\begin{tabular}{@{}lcc@{}}
\toprule
\textbf{Metric} & \textbf{Median value} & \textbf{Best}\\
\midrule
Qubits                                         & 127 & -- \\
Two-qubit ECR error $\epsilon_{2Q}$            & $6.99\times10^{-3}$ & $2.67\times10^{-3}$ \\
Readout error                                  & $1.61\times10^{-2}$ & $5.1\times10^{-3}$ \\
$T_{1}$, $T_{2}$ (\(\mu\)s)                    & 235.7, 122.4 & -- \\
\bottomrule
\end{tabular}
\end{center}

The median two-qubit error rate lies in the \(10^{-3}\) to \(10^{-2}\) regime,
so these runs probe CE--QAOA well outside an idealized noiseless setting. The
purpose of the experiment is therefore not to test perfect coherent evolution,
but to assess whether shallow constrained sampling combined with deterministic
repair remains useful under realistic device noise.

\paragraph{Transpilation statistics.}
Table~\ref{tab:brisbane-depth} reports the transpiled depth, two-qubit gate
count, and turnaround time obtained on \texttt{ibm\_brisbane}. Dynamical
decoupling and Pauli twirling were enabled in the reported transpilation
pipeline, and the resulting circuits reflect the cost of embedding the encoded
CE--QAOA structure into the heavy-hex connectivity.

\begin{table}[H]
\centering
\caption{Logical qubit count (\(n^{2}\)), two-qubit gates,
physical depth, and end-to-end turnaround per instance on
\texttt{ibm\_brisbane}. Wall time includes queue and execution overhead.}
\label{tab:brisbane-depth}
\begin{tabular}{lrrrrr}
\toprule
Instance & Qubits & 2Q gates & Depth & Wall [s] \\
\midrule
wi4  & 16  &   509  &   771  &  761 \\
wi5  & 25  & 1\,387  & 1\,787 &   80 \\
wi6  & 36  & 2\,834  & 3\,060 &  137 \\
wi7  & 49  & 5\,141  & 4\,706 &  228 \\
dj8  & 64  & 8\,295  & 7\,221 &  317 \\
dj9  & 81  & 12\,642 & 10\,695 &  499 \\
dj10 & 100 & 18\,770 & 12\,962 &  719 \\
\bottomrule
\end{tabular}
\end{table}

\subsection{Hardware Discussion.}
\begin{itemize}[leftmargin=*]
\item The transpiled circuits are deep and strongly noise-affected, especially for the \texttt{dj} instances.

\item Hungarian repair restores feasibility deterministically for every sampled
      outcome, so the post-processed output is always a valid permutation.

\item The local classical post-processing cost is negligible compared with the
      queue and hardware execution overhead, even though every sample is
      repaired separately.

\item  For the benchmark set
      considered here, the best repaired tours are competitive with the
      published \textsc{QOptlib} references and, for  the harder
      instances, improve on them.
\end{itemize}

\begin{figure}[H]
  \centering
  \includegraphics[width=.9\linewidth]{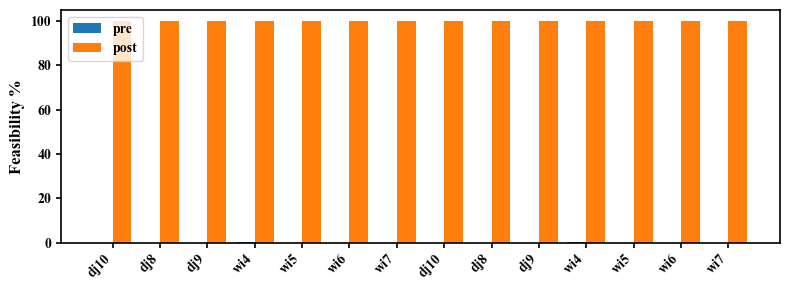}%
  \caption{Raw versus post-repair feasible fraction for the
           \textsc{QOptlib} benchmark instances executed on
           \texttt{ibm\_brisbane}.  Hungarian repair maps all samples back into the
           feasible permutation set.}
  \label{fig:brisbane-feasible}
\end{figure}

\vspace{-1ex}

\subsubsection{Comparison with published \textsc{QOptlib} tours}

Table~\ref{tab:qoptlib-comparison} compares the repaired CE--QAOA output on
\texttt{ibm\_brisbane} against the reference tours distributed with
\textsc{QOptlib} \cite{Osaba2024Qoptlib}. The smaller \texttt{wi} instances are
matched exactly, while the harder \texttt{dj} subset shows modest to
double-digit improvements according to the reported fixed-shot comparison.

\begin{table}[H]
\centering
\caption{Repaired CE--QAOA output on \texttt{ibm\_brisbane} versus the
reference tours distributed with \textsc{QOptlib}. Positive \% improvement means a
shorter tour than the published reference.}
\label{tab:qoptlib-comparison}
\begin{tabular}{lrrrr}
\toprule
Instance & QOPTLib cost & Our cost & $\Delta$ & Improvement [\%] \\
\midrule
wi4  & 6\,700 & 6\,700 & 0  & 0.00 \\
wi5  & 6\,786 & 6\,786 & 0  & 0.00 \\
wi6  & 9\,815 & 9\,815 & 0  & 0.00 \\
wi7  & 7\,245 & 7\,245 & 0  & 0.00 \\
dj8  & 2\,794 & 2\,762 & $-32$  & \textbf{1.15} \\
dj9  & 2\,438 & 2\,134 & $-304$ & \textbf{12.5} \\
dj10 & 3\,155 & 2\,822 & $-333$ & \textbf{10.55} \\
\bottomrule
\end{tabular}
\end{table}

\noindent
Taken at face value, the standardized comparison indicates that even in a
strongly noisy regime the repaired CE--QAOA samples remain competitive with the reference benchmark tours.

\subsection{\texttt{ibm\_sherbrooke} (127 qubits, Eagle r3)}
\label{ssec:exp-sherbrooke}

\paragraph{Device snapshot (21 June 2025, 07:10 UTC).}
\vspace{-0.5em}
\begin{center}
\begin{tabular}{@{}lcc@{}}
\toprule
\textbf{Metric} & \textbf{Median} & \textbf{Best} \\
\midrule
Two-qubit ECR error $\epsilon_{2Q}$          & $7.13\times10^{-3}$ & $2.55\times10^{-3}$ \\[1pt]
Read-out error                               & $2.08\times10^{-2}$ & $3.9\times10^{-3}$  \\[1pt]
Single-qubit SX error                        & $2.27\times10^{-4}$ & $6.8\times10^{-5}$  \\[1pt]
$T_{1}$, $T_{2}$ coherence (\si{\micro\second}) & 262, 212 & – \\[1pt]
Calib.\ version                              & \multicolumn{2}{c}{1.6.112 (Eagle r3)} \\
\bottomrule
\end{tabular}
\end{center}

\paragraph{Transpilation statistics.}
\begin{table}[H]
\centering
\caption{Logical-qubit count (\(n^{2}\)), two-qubit ECR gates,
physical depth, and end-to-end turnaround on \texttt{ibm\_sherbrooke}.}
\label{tab:sherbrooke-depth}
\begin{tabular}{lrrrrr}
\toprule
Instance & Qubits & 2Q gates & Depth & Wall [s] \\
\midrule
wi4  & 16  &   547  &   907  &  181 \\
wi5  & 25  & 1\,477 & 1\,964 &   74 \\
wi6  & 36  & 2\,876 & 3\,031 &  150 \\
wi7  & 49  & 5\,060 & 4\,842 &  689 \\
dj8  & 64  & 8\,507 & 6\,896 &  551 \\
dj9  & 81  & 12\,827 & 9\,524 &  574 \\
dj10 & 100 & 18\,768 & 11\,698 &  728 \\
\bottomrule
\end{tabular}
\end{table}

\begin{figure}[H]
  \centering
  \includegraphics[width=.9\linewidth]{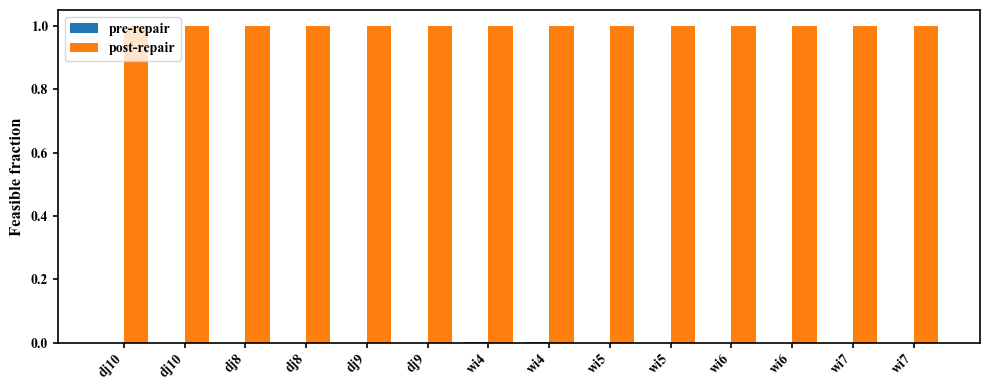}%
  \caption{Raw versus post-repair feasible fraction for
           \texttt{ibm\_sherbrooke}. As on \texttt{ibm\_brisbane}, raw
           feasibility is very small and deterministic repair restores all
           outputs to feasibility.}
  \label{fig:sherbrooke-feasible}
\end{figure}

\begin{table}[H]
\centering
\caption{Best repaired tours obtained on \texttt{ibm\_sherbrooke} in the
standardized comparison, together with the \textsc{QOptlib} reference tours
\cite{Osaba2024Qoptlib}.}
\label{tab:sherbrooke-comparison}
\begin{tabular}{lrrrrr}
\toprule
Instance & Qubits & Our cost & QOPTLib cost & $\Delta$ & Improvement [\%] \\
\midrule
dj10 & 100 & 2\,909 & 3\,155 & $-246$ & \textbf{7.8}  \\
dj9  &  81 & 2\,134 & 2\,438 & $-304$ & \textbf{12.5} \\
dj8  &  64 & 2\,762 & 2\,794 &  $-32$ & \textbf{1.1}  \\
wi7  &  49 & 7\,245 & 7\,245 &   0    & 0.0  \\
wi6  &  36 & 9\,815 & 9\,815 &   0    & 0.0  \\
wi5  &  25 & 6\,786 & 6\,786 &   0    & 0.0  \\
wi4  &  16 & 6\,700 & 6\,700 &   0    & 0.0  \\
\bottomrule
\end{tabular}
\end{table}

The \texttt{ibm\_sherbrooke} runs are qualitatively similar to the
\texttt{ibm\_brisbane} results. The transpiled circuits are again deep and
strongly noise-affected; raw feasibility is extremely small; and deterministic
Hungarian repair maps every sample back to a valid permutation. Despite this,
the repaired samples remain competitive with the benchmark references, matching
the known \texttt{wi} solutions and improving several \texttt{dj} instances in
the standardized comparison.

% A point of reference is the combinatorial baseline
% \[
%   \frac{n!}{n^n},
% \]
% the fraction of valid permutations inside the row-one-hot set. This quantity
% already falls below \(10^{-3}\) for the larger instances. Thus the observed
% pre-repair feasibility of \(0\) to \(0.1\%\) is not surprising even before
% hardware noise is taken into account. Device noise then further suppresses raw
% feasible mass by scattering weight outside the intended constrained sector. The two experiments show that although raw output quality is poor, deterministic repair turns shallow constrained samples into useful feasible candidates, and the resulting repaired solutions remain competitive with the
% benchmark references.

\newcommand{\HH}{\mathcal H}

\section{HH-QAOA: Geometry Reduced Shot Analyses Preserving Guarantees}

We propose Heavy-Hitter QAOA, a Heavy-Hitter theory that leverages geometry and quantum information to equip the NP-HQ pipeline with a memory-efficient and time-efficient classical post-processing scheme. The key idea is that
the classical checker can focus on the small set of bit-strings that received the highest
sampling bias from the quantum core. Under purely classical heavy-hitter
bounds the retained candidate set has size $O(n^{k+1})$, independent of
the raw shot count, leading to post-processing cost $O(n^{k+4})$. By
exploiting the inverse-polynomial optimum mass guaranteed by the
CE--QAOA kernel, the admissible threshold can be raised by one factor
of~$n$, reducing the candidate set to $O(n^k)$ and the post-processing
cost to $O(n^{k+3})$. This reduction comes without weakening the
approximation guarantees established in the preceding sections.

%----------------------------------------------------------------------
\begin{definition}[Heavy hitter]\label{def:heavy-hitter}
Fix integers \(n,k\ge 1\) and let \(\tau=\Theta\!\bigl(n^{-(k+1)}\bigr)\).
For a probability distribution \(P\) on \(\{0,1\}^{n}\),
a bit‑string \(y\in\{0,1\}^{n}\) is called a \emph{heavy hitter}
if its probability mass satisfies \(P(y)\ge \tau\).
Empirically, given \(S\) samples \(\{y^{(s)}\}_{s=1}^{S}\),
\(y\) is a heavy hitter when its observed frequency obeys
\(\widehat P_{S}(y)\ge \tau\).
\end{definition}

%============================================================
\subsection{Speeding up the Analyses of Quantum-Generated Samples}
\label{sec:hh-from-circuits}
%============================================================

\paragraph{Relation to the exact-hit regime.}
We propose heavy-hitter refinement as a post-processing
refinement of NP-HQ algorithm. Indeed, when the optimal mass satisfies
\(p_{\min}=\Theta(n^{-k})\), the exact-hit shot budget is already
\(S=\Theta(n^k\log(1/\delta))\). In that regime, heavy-hitter thresholding
does not change the fundamental success exponent controlling the total number
of required samples. Its role is instead to organize the retained candidate
set more efficiently and to expose additional practical reductions in memory
and classical workload, especially in oversampled and hardware-scale settings
where the raw histogram is much larger than the minimal exact-hit budget.
Accordingly, the theorem-level guarantees in this section should be read as
refinements of post-processing, while the stronger reductions seen under
block-prefiltering and top-\(L\) compression are empirical or
instance-dependent. Yet, as we argue below, most of that support is irrelevant for
post‑processing efforts to identify the global optimum because a set of \emph{heavy hitters} captures essentially
all the useful weight containing it. First we invoke classical heavy‑hitter theory to quantify this
intuition for arbitrary distributions. Then in CE‑QAOA we leverage additional
algebraic, geometric and information‑theoretic structure to shrink the
heavy‑hitter list even further, culminating in what we call a
\emph{Quantum Heavy‑Hitter Theory} (Sec.~\ref{ssec:q-hh-theory}). 

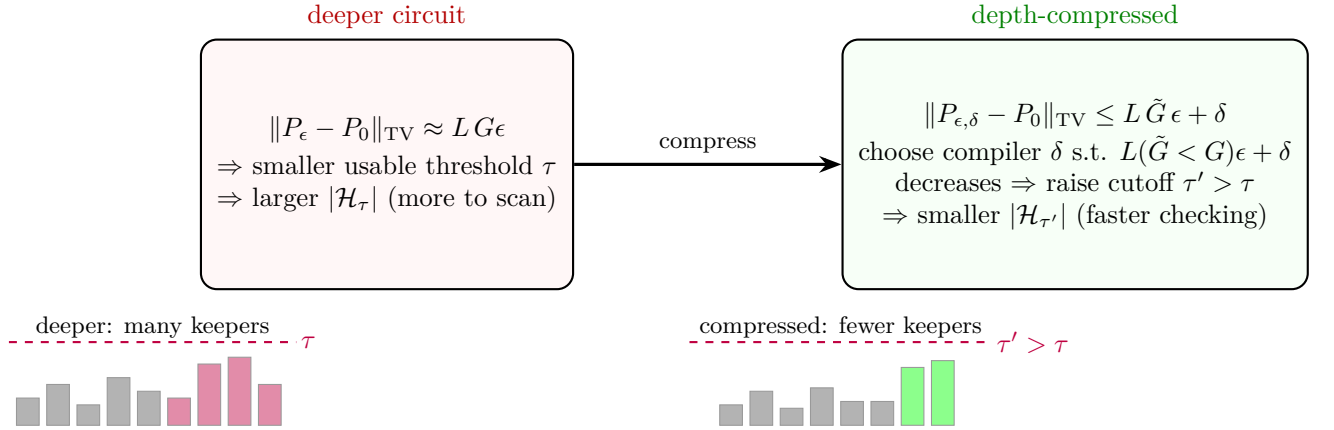
\begin{figure}[H]
\centering
\begin{tikzpicture}[
  font=\small, >=Stealth,
  box/.style={draw, rounded corners=6pt, thick, inner sep=6pt, minimum width=3.3cm, minimum height=3.3cm, align=center}
]
% Left
\node[box, fill=red!3, label={[red!70!black]above:deeper circuit}] (D) at (0,0)
{\(\|P_{\epsilon}-P_0\|_{\rm TV}\approx L\,G\epsilon\)\\[2pt]
\(\Rightarrow\) smaller usable threshold \(\tau\)\\
\(\Rightarrow\) larger \(|\mathcal H_\tau|\) (more to scan)};

% Right
\node[box, fill=green!3, label={[green!50!black]above:depth-compressed}] (C) at (9.1,0)
{\(\|P_{\epsilon,\delta}-P_0\|_{\rm TV}\le L\,\tilde G\,\epsilon + \delta\)\\[2pt]
choose compiler \(\delta\) s.t. \(L(\tilde G<G)\epsilon+\delta\)\\ decreases \(\Rightarrow\) raise cutoff \(\tau'>\tau\)\\
\(\Rightarrow\) smaller \(|\mathcal H_{\tau'}|\) (faster checking)};

\draw[-Stealth, very thick] (D.east) -- node[above]{\footnotesize compress} (C.west);

% Bottom histograms
\begin{scope}[shift={(-3.1,-2.15)}]
  \node at (0,0) {\footnotesize deeper: many keepers};
  \foreach \i/\h/\col in {0/0.4/gray!60,1/0.6/gray!60,2/0.3/gray!60,3/0.7/gray!60,4/0.5/gray!60,5/0.4/purple!45,6/0.9/purple!45,7/1.0/purple!45,8/0.6/purple!45}{
    \draw[fill=\col, draw=black!40] (-1.8+0.4*\i,-1.3) rectangle ++(0.3,\h*0.9);
  }
  \draw[thick, purple, dashed] (-1.9,-0.2) -- (1.9,-0.2) node[pos=0.98, right] {\(\tau\)};
\end{scope}

\begin{scope}[shift={(6.0,-2.15)}]
  \node at (0,0) {\footnotesize compressed: fewer keepers};
  \foreach \i/\h/\col in {0/0.3/gray!60,1/0.5/gray!60,2/0.25/gray!60,3/0.55/gray!60,4/0.35/gray!60,5/0.35/gray!60,6/0.85/green!45,7/0.95/green!45}{
    \draw[fill=\col, draw=black!40] (-1.6+0.4*\i,-1.3) rectangle ++(0.3,\h*0.9);
  }
  \draw[thick, purple, dashed] (-2.0,-0.2) -- (2.0,-0.2) node[pos=0.98, right] {\(\tau'>\tau\)};
\end{scope}

\end{tikzpicture}
\caption{\textbf{Depth-compression effect.} Trading a tiny compiler error \(\delta\) for a large gate-count reduction \(\tilde G<G\) reduces total TV error, lets you raise \(\tau\), and shrinks the heavy-hitter set. The associated circuit compression is discussed in App. \ref{sec:depth-planning}}
\label{fig:depth_compression_tradeoff}
\end{figure}

%========================================================
% Figure 4
%========================================================
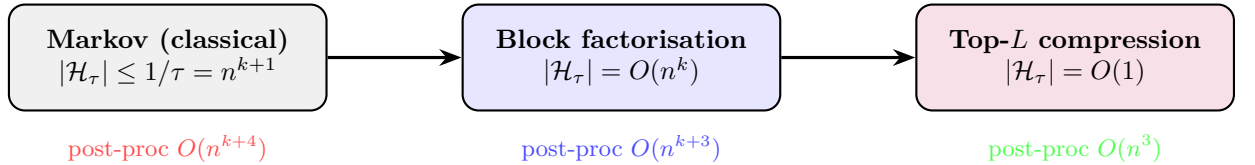
\begin{figure}[H]
\centering
\begin{tikzpicture}[
  font=\small,
  boxnode/.style={draw, rounded corners=6pt, thick, align=center, minimum width=4.2cm, minimum height=1.4cm},
  arr/.style={-Stealth, very thick}
]
\node[boxnode, fill=gray!12]   (A) at (0,0)    {\textbf{Markov (classical)}\\ $|\mathcal H_\tau|\le 1/\tau = n^{k+1}$};
\node[boxnode, fill=blue!10]   (B) at (6.0,0)  {\textbf{Block factorisation}\\ $|\mathcal H_\tau|=O(n^{k})$};
\node[boxnode, fill=purple!12] (C) at (12.0,0) {\textbf{Top-\(L\) compression}\\ $|\mathcal H_\tau|=O(1)$};

\draw[arr] (A) -- (B);
\draw[arr] (B) -- (C);

\node[align=left, red!70] at (0,-1.2)    {\footnotesize post-proc $O(n^{k+4})$};
\node[align=left, blue!70] at (6.0,-1.2) {\footnotesize post-proc $O(n^{k+3})$};
\node[align=left, green!70] at (12.0,-1.2) {\footnotesize post-proc $O(n^{3})$};

\end{tikzpicture}
\caption{\textbf{Shrinking the scan set.} Structure in CE--QAOA collapses the heavy-hitter list from \(n^{k+1}\) to \(O(1)\), cutting classical time from \(O(n^{k+4})\) to \(O(n^{3})\).}
\label{fig:hh_set_size_ladder}
\end{figure}

%------------------------------------------------------------
\subsection{Classical Heavy-Hitter Theory}
\label{ssec:classical-HH}
%------------------------------------------------------------

\paragraph{Setup.}
Let \(p\) be any probability distribution supported on a finite feasible set
\(\Omega\) of size \(D\). For a threshold \(\tau>0\), define the heavy-hitter set
\[
  \mathcal H_\tau
  \;:=\;
  \{x\in\Omega : p(x)\ge \tau\},
  \qquad
  H_\tau
  \;:=\;
  |\mathcal H_\tau|.
\]

\begin{definition}[Heavy-hitter set]
\label{def:HH}
For a distribution \(p\) on \(\Omega\) and threshold \(\tau>0\), the
\emph{heavy-hitter set} is
\[
  \mathcal H_\tau
  \;:=\;
  \{x\in\Omega : p(x)\ge \tau\}.
\]
\end{definition}

\begin{lemma}[Cardinality bound]
\label{lem:HH-cardinality}
For every distribution \(p\) and every threshold \(\tau>0\),
\[
  H_\tau \;\le\; \frac{1}{\tau}.
\]
\end{lemma}

\begin{proof}
By definition, each \(x\in\mathcal H_\tau\) contributes at least \(\tau\) mass, so
\[
  1
  \;=\;
  \sum_{x\in\Omega} p(x)
  \;\ge\;
  \sum_{x\in\mathcal H_\tau} p(x)
  \;\ge\;
  H_\tau\,\tau.
\]
Hence \(H_\tau\le 1/\tau\).
\end{proof}

\begin{lemma}[Empirical retention of heavy hitters]
\label{lem:HH-retention}
Let \(X_1,\dots,X_S\) be i.i.d.\ samples from \(p\), and let
\[
  \widehat p_S(x)
  \;:=\;
  \frac{1}{S}\sum_{s=1}^S \mathbf 1\{X_s=x\}
\]
be the empirical frequency.
Define the empirical candidate set
\[
  \widehat{\mathcal H}_{\tau/2}
  \;:=\;
  \{x\in\Omega : \widehat p_S(x)\ge \tau/2\}.
\]
Then the following hold.

\begin{enumerate}[label=\textup{(\roman*)},leftmargin=*]
\item For every \(x\in\mathcal H_\tau\),
\[
  \Pr\!\left[\widehat p_S(x)<\frac{\tau}{2}\right]
  \;\le\;
  \exp\!\left(-\frac{\tau S}{8}\right).
\]

\item Consequently,
\[
  \Pr\!\left[\mathcal H_\tau \subseteq \widehat{\mathcal H}_{\tau/2}\right]
  \;\ge\;
  1-\frac{1}{\tau}\exp\!\left(-\frac{\tau S}{8}\right).
\]

\item In every sample realization,
\[
  |\widehat{\mathcal H}_{\tau/2}|
  \;\le\;
  \frac{2}{\tau}.
\]
\end{enumerate}
In particular, if
\[
  S
  \;\ge\;
  \frac{8}{\tau}\log\!\frac{1}{\tau\delta},
\]
then with probability at least \(1-\delta\),
\[
  \mathcal H_\tau \subseteq \widehat{\mathcal H}_{\tau/2}
  \quad\text{and}\quad
  |\widehat{\mathcal H}_{\tau/2}|\le \frac{2}{\tau}.
\]
\end{lemma}

\begin{proof}
For \(x\in\mathcal H_\tau\), the count
\(
N_x:=\sum_{s=1}^S \mathbf 1\{X_s=x\}
\)
is binomial with mean
\(
\mathbb E[N_x]=S\,p(x)\ge S\tau.
\)
By the multiplicative Chernoff bound,
\[
  \Pr\!\left[N_x<\frac{1}{2}\mathbb E[N_x]\right]
  \;\le\;
  \exp\!\left(-\frac{\mathbb E[N_x]}{8}\right)
  \;\le\;
  \exp\!\left(-\frac{\tau S}{8}\right),
\]
which proves (i). Since \(|\mathcal H_\tau|\le 1/\tau\) by Lemma~\ref{lem:HH-cardinality},
a union bound gives (ii).

For (iii), every \(x\in \widehat{\mathcal H}_{\tau/2}\) contributes at least
\(\tau/2\) empirical mass, so
\[
  1
  \;=\;
  \sum_{x\in\Omega} \widehat p_S(x)
  \;\ge\;
  \sum_{x\in\widehat{\mathcal H}_{\tau/2}} \widehat p_S(x)
  \;\ge\;
  |\widehat{\mathcal H}_{\tau/2}|\,\frac{\tau}{2}.
\]
Hence \(|\widehat{\mathcal H}_{\tau/2}|\le 2/\tau\).
The final statement follows by choosing \(S\) so that
\(
\tau^{-1}e^{-\tau S/8}\le \delta
\).
\end{proof}

\begin{remark}[Exact recovery requires a margin]
\label{rem:HH-margin}
Without an explicit separation gap around the threshold \(\tau\), one should not
claim exact recovery of \(\mathcal H_\tau\) from finitely many samples. What
Lemma~\ref{lem:HH-retention} guarantees is the more robust statement that all
true \(\tau\)-heavy strings are retained inside a slightly relaxed empirical set
of size \(O(1/\tau)\). Exact identification of the thresholded set requires an
additional margin assumption, for example
\[
  p(x)\ge \tau+\eta \quad \text{for } x\in\mathcal H_\tau,
  \qquad
  p(x)\le \tau-\eta \quad \text{for } x\notin\mathcal H_\tau,
\]
together with a Hoeffding/union-bound argument over \(\Omega\).
\end{remark}

\begin{corollary}[Classical heavy-hitter post-processing cost]
\label{cor:HH-runtime}
Suppose the deterministic checker/repair/scoring stage costs \(O(n^3)\) per
candidate. If the empirical candidate set is chosen as
\(
\widehat{\mathcal H}_{\tau/2}
\),
then with probability at least \(1-\delta\) all true \(\tau\)-heavy strings are
retained and the number of candidates is at most \(2/\tau\). Hence the total
post-processing cost is
\[
  O\!\left(\frac{n^3}{\tau}\right).
\]
In particular, for the choice
\[
  \tau = n^{-(k+1)},
\]
the retained set has size \(O(n^{k+1})\) and the post-processing cost is
\[
  O(n^{k+4}).
\]
\end{corollary}

\begin{proposition}[Worst-case sharpness of the cardinality bound]
\label{prop:HH-sharp}
The bound \(H_\tau\le 1/\tau\) is worst-case sharp up to rounding. More precisely,
for every \(\tau\in(0,1]\), there exists a distribution \(p\) such that
\[
  H_\tau = \left\lfloor \frac{1}{\tau}\right\rfloor .
\]
\end{proposition}

\begin{proof}
Let \(N=\lfloor 1/\tau\rfloor\) and define \(p\) to be uniform on \(N\) points:
\[
  p(x)=\frac{1}{N}
  \qquad \text{for } x\in\{x_1,\dots,x_N\}.
\]
Since \(N\le 1/\tau\), \(1/N\ge \tau\), so all \(N\) support points are
\(\tau\)-heavy. Thus \(H_\tau=N=\lfloor 1/\tau\rfloor\).
\end{proof}

%------------------------------------------------------------
\subsection{Quantum Heavy-Hitter Refinements}
\label{ssec:q-hh-theory}
%------------------------------------------------------------

The classical heavy-hitter bound of Sec.~\ref{ssec:classical-HH} applies to
\emph{any} output distribution and therefore ignores the additional structure
present in CE--QAOA.  In the present setting, a quantum-informed threshold  uses the inverse-polynomial
      optimum weight proved for the CE kernel together with the TV-noise bound
      to raise the heavy-hitter cutoff from the generic scale
      \(n^{-(k+1)}\) to the sharper scale \(n^{-k}\). So the main refinement comes when one connects the lower bound on the ideal optimum weight and a total-variation stability bound under noise. These two ingredients
allow one to choose a \emph{larger} heavy-hitter threshold than in the generic classical analysis. For this refinement, we impose the stronger condition that at least one optimal string carries inverse-polynomial
probability mass and state the result as follows.

\begin{theorem}[Quantum-informed noisy heavy hitters]
\label{thm:quantum-informed-HH}
Let \(x^\star\) denote the optimal feasible bit-string. Assume the ideal CE--QAOA
output distribution \(p_0\) satisfies
\[
  p_0(x^\star)\;\ge\; c\,n^{-k}
\]
for fixed constants \(c>0\) and \(k\ge 0\), and let \(p_\epsilon\) be the noisy
output distribution. Suppose further that
\[
  \|p_\epsilon - p_0\|_{\mathrm{TV}}
  \;\le\;
  \Delta_n,
  \qquad
  \Delta_n \le \frac{c}{4}\,n^{-k}.
\]
Define the quantum-informed threshold
\(
  \tau_q
  \;:=\;
  \frac{c}{4}\,n^{-k}.
\)
Then:
\begin{enumerate}[label=\textup{(\roman*)},leftmargin=*]
\item the optimal string remains heavy under noise,
\(
  p_\epsilon(x^\star)\;\ge\;\tau_q;
\)
\item the noisy heavy-hitter set at threshold \(\tau_q\),
\(
  \mathcal H_{\tau_q}(p_\epsilon)
  \;:=\;
  \{x: p_\epsilon(x)\ge \tau_q\},
\)
has cardinality
\(
  |\mathcal H_{\tau_q}(p_\epsilon)|
  \;\le\;
  \frac{1}{\tau_q}
  \;=\;
  \frac{4}{c}\,n^k.
\)
\end{enumerate}
\end{theorem}

\begin{proof}
For the singleton event \(\{x^\star\}\), total variation distance gives
\[
  |p_\epsilon(x^\star)-p_0(x^\star)|
  \;\le\;
  \|p_\epsilon-p_0\|_{\mathrm{TV}}
  \;\le\;
  \Delta_n.
\]
Hence
\[
  p_\epsilon(x^\star)
  \;\ge\;
  p_0(x^\star)-\Delta_n
  \;\ge\;
  c\,n^{-k} - \frac{c}{4}\,n^{-k}
  \;=\;
  \frac{3c}{4}\,n^{-k}
  \;\ge\;
  \tau_q,
\]
which proves (i). Part (ii) is Lemma~\ref{lem:HH-cardinality} applied with
\(\tau=\tau_q\).
\end{proof}

\begin{corollary}[Empirical candidate list at quantum-informed threshold]
\label{cor:quantum-informed-empirical}
Under the hypotheses of Theorem~\ref{thm:quantum-informed-HH}, define the
empirical candidate set
\[
  \widehat{\mathcal H}_{\tau_q/2}
  \;:=\;
  \{x : \widehat p_S(x)\ge \tau_q/2\}.
\]
If
\(
  S
  \;\ge\;
  \frac{8}{\tau_q}\log\!\frac{1}{\tau_q\delta}
  \;=\;
  O\!\bigl(n^k \log(n/\delta)\bigr),
\)
then with probability at least \(1-\delta\),
\[
  x^\star \in \widehat{\mathcal H}_{\tau_q/2}
  \qquad\text{and}\qquad
  |\widehat{\mathcal H}_{\tau_q/2}|
  \;\le\;
  \frac{2}{\tau_q}
  \;=\;
  \frac{8}{c}\,n^k.
\]
\end{corollary}

\begin{proof}
Apply Lemma~\ref{lem:HH-retention} from the classical heavy-hitter subsection
with \(\tau=\tau_q\), together with Theorem~\ref{thm:quantum-informed-HH},
which ensures that \(x^\star\in \mathcal H_{\tau_q}(p_\epsilon)\).
\end{proof}

The generic classical heavy-hitter analysis used a threshold of order
\(n^{-(k+1)}\), leading to a retained set of size \(O(n^{k+1})\).
Theorem~\ref{thm:quantum-informed-HH} uses the kernel guarantee to obtain the sharper threshold \(\tau_q=\Theta(n^{-k})\). This immediately improves the
retained-set size by one factor of \(n\), from \(O(n^{k+1})\) to \(O(n^k)\),
while preserving the optimum with high probability after polynomially many shots.

\subsection{Heavy-Hitter Extraction and Runtime Gains}
\label{ssec:hh-algorithms}
\label{ssec:runtime-gains}

For the empirical evaluation, we supplement the theorem-level
heavy-hitter threshold with two frequency-ranked compression steps.
After forming the global empirical heavy-hitter set
\(\mathcal C_{\mathrm{HH}}\), we apply an empirical
\emph{block filter} by retaining its \(n^k\) most frequent elements:
\[
\mathcal C_{\mathrm{blk}}
:=
\operatorname{Top}_{n^k}(\mathcal C_{\mathrm{HH}}).
\]
We may then retain only the \(L\) most frequent elements of the
block-filtered set:
\[
\mathcal C_L
:=
\operatorname{Top}_L(\mathcal C_{\mathrm{blk}}).
\]

If deterministic checking, repair, and
scoring cost \(O(n^3)\) per retained candidate, then
\[
T_{\mathrm{post}}
=
O\!\bigl(|\mathcal C|n^3\bigr),
\]
where \(\mathcal C\) is the retained candidate set.  The first pipeline
uses only the generic classical heavy-hitter bound of
Sec.~\ref{ssec:classical-HH}. The second uses the sharper
\emph{quantum-informed} threshold from
Theorem~\ref{thm:quantum-informed-HH}.

\subsubsection{Empirical Footprints of Heavy-Hitter Filtering}
\label{sec:hh-robustness}

We now examine the empirical footprint of the three post-processing stages
introduced in Sec.~\ref{ssec:hh-algorithms}: the \emph{classical heavy-hitter filter}, the \emph{block filtering stage}, and the optional \emph{top-\(L\)} empirical compression. The purpose of this section is to quantify the practical reduction in
candidate-set size observed on representative instances. For each run, we
report the retained candidate count after the classical thresholding stage,
the count after block-prefiltering, and the count after the optional top-\(L\)
compression step.  We also record the percentage cost gap relative to the best
solution obtained from the full retained histogram.

\begin{table}[ht]
\caption{Empirical candidate-set footprints and percentage cost gaps for \(k=3\) with top-\(L\) window \(L=25k=75\). Here \(|\mathcal C_{\mathrm{HH}}|\) is the retained set after classical thresholding, \(|\mathcal C_{\mathrm{blk}}|\) is the retained set after block-prefiltering, and \(|\mathcal C_{L}|\) is the final top-\(L\) compressed set.}
\label{tab:hh-gaps-k3-c25}
\centering
\begin{tabular}{lrrrrrrr}
\toprule
Instance & Shots & \(|\mathcal C_{\mathrm{HH}}|\) & \(|\mathcal C_{\mathrm{blk}}|\) & \(|\mathcal C_{L}|\)
& Gap HH [\%] & Gap blk [\%] & Gap top-\(L\) [\%] \\
\midrule
eil10  & 5000   & 5000  & 5000 & 75 & 0.00 & 0.00 & 0.00 \\
dj8    & 10000  & 10000 & 512  & 75 & 0.00 & 0.07 & 0.07 \\
dj9    & 10000  & 10000 & 729  & 75 & 0.00 & 0.19 & 0.19 \\
wi4    & 10000  & 128   & 64   & 64 & 0.00 & 0.00 & 0.00 \\
wi4    & 100000 & 234   & 64   & 64 & 0.00 & 0.00 & 0.00 \\
wi5    & 10000  & 6     & 6    & 6  & 0.52 & 0.52 & 0.52 \\
wi5    & 100000 & 2449  & 125  & 75 & 0.00 & 0.00 & 0.00 \\
wi6    & 10000  & 10000 & 216  & 75 & 0.00 & 0.00 & 0.00 \\
wi7    & 10000  & 10000 & 343  & 75 & 0.00 & 0.04 & 0.04 \\
\bottomrule
\end{tabular}
\end{table}

\begin{table}[ht]
\caption{Empirical candidate-set footprints and \% cost gaps for \(k=4\) with top-\(L\) window \(L=25k=100\).}
\label{tab:hh-gaps-k4-c25}
\centering
\begin{tabular}{lrrrrrrr}
\toprule
Instance & Shots & \(|\mathcal C_{\mathrm{HH}}|\) & \(|\mathcal C_{\mathrm{blk}}|\) & \(|\mathcal C_{L}|\)
& Gap HH [\%] & Gap blk [\%] & Gap top-\(L\) [\%] \\
\midrule
eil10  & 5000   & 5000  & 5000 & 100 & 0.00 & 0.00 & 0.00 \\
dj8    & 10000  & 10000 & 4096 & 100 & 0.00 & 0.00 & 0.07 \\
dj9    & 10000  & 10000 & 6561 & 100 & 0.00 & 0.00 & 0.19 \\
wi4    & 10000  & 1398  & 256  & 100 & 0.00 & 0.00 & 0.00 \\
wi4    & 100000 & 253   & 253  & 100 & 0.00 & 0.00 & 0.00 \\
wi5    & 10000  & 9994  & 625  & 100 & 0.00 & 0.00 & 0.00 \\
wi5    & 100000 & 3113  & 625  & 100 & 0.00 & 0.00 & 0.00 \\
wi6    & 10000  & 10000 & 1296 & 100 & 0.00 & 0.00 & 0.00 \\
wi7    & 10000  & 10000 & 2401 & 100 & 0.00 & 0.00 & 0.04 \\
\bottomrule
\end{tabular}
\end{table}

Tables~\ref{tab:hh-gaps-k3-c25} and~\ref{tab:hh-gaps-k4-c25} show that the
block-prefiltering step can reduce the retained set substantially relative to
the purely classical thresholding stage, especially on the larger benchmark
instances. The optional top-\(L\) compression produces a further reduction in
footprint, often with negligible change in the best recovered cost. This should
be interpreted as an empirical acceleration rather than a formal guarantee:
its effectiveness depends on how sharply peaked the filtered histogram is.

To probe robustness more directly, we ran multiple trials on the
\texttt{wi4}--\texttt{wi7} instances with varying shot budgets. In each trial we
recorded the retained set size after classical thresholding, after
block-prefiltering, and after optional top-\(L\) compression, together with the
cost gap relative to the best solution obtained from the full retained set.

The main empirical lesson is that the block-prefiltering stage is the most reliable source of compression. It
consistently reduces the candidate-set size while preserving the best recovered
cost in all reported runs. The additional top-\(L\) compression is much more
aggressive and usually remains harmless, but it is best viewed as a practical
accelerator rather than part of the formal guarantee. The runtime gains are plotted in Fig~\ref{fig:empirical-runtimes}.  

\begin{figure}[htbp]
  \centering
  \begin{subfigure}[t]{0.32\textwidth}
    \centering
    \includegraphics[width=\linewidth]{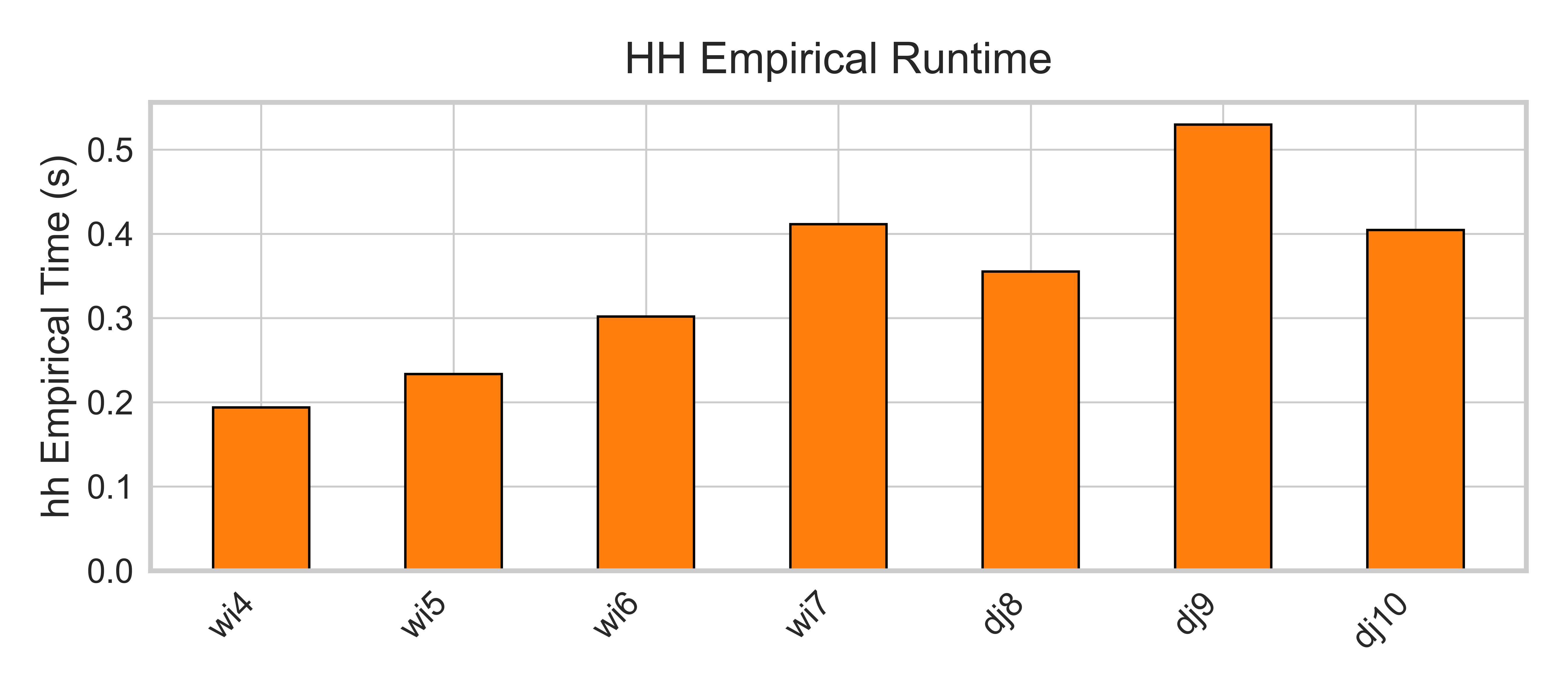}
    \caption{Empirical runtime of the classical heavy-hitter filter.}
    \label{fig:hh-runtime}
  \end{subfigure}
  \hfill
  \begin{subfigure}[t]{0.32\textwidth}
    \centering
    \includegraphics[width=\linewidth]{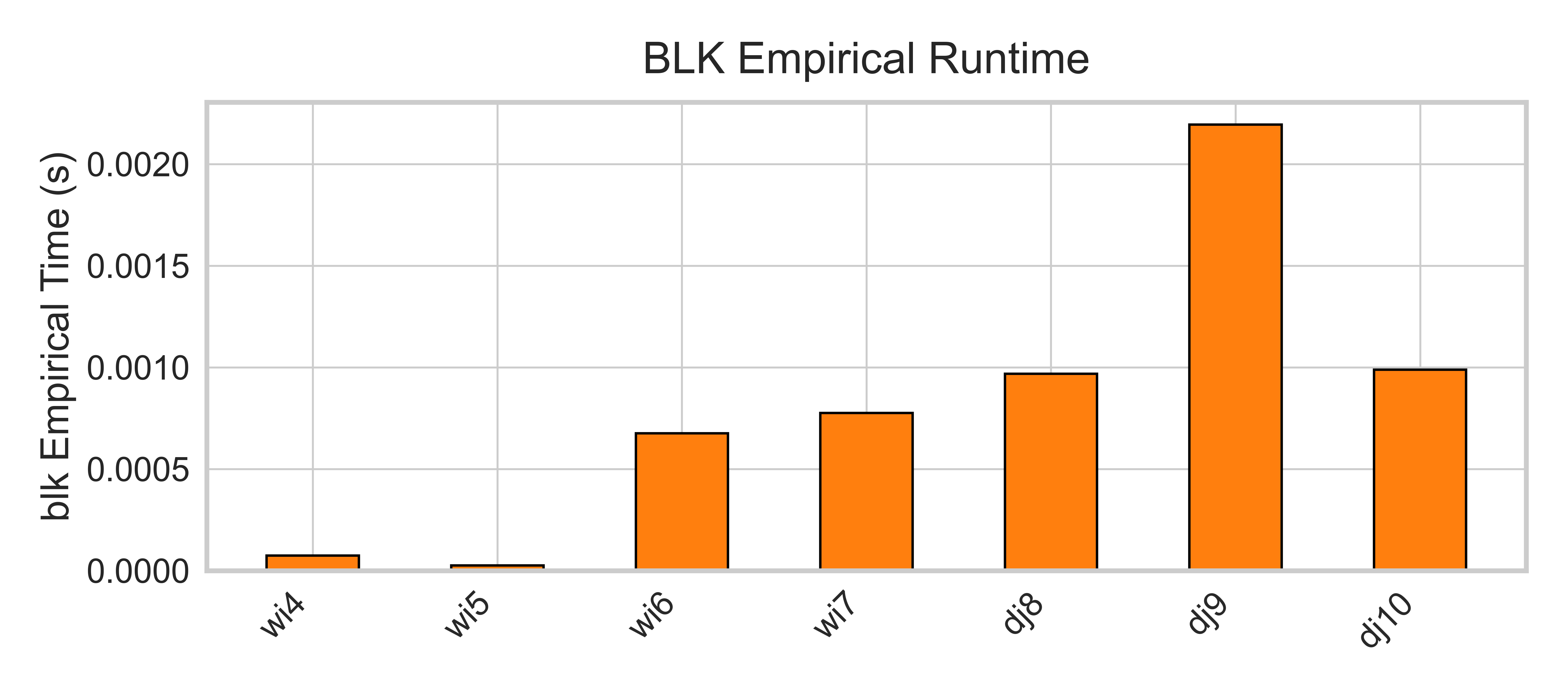}
    \caption{Empirical runtime after block-prefiltering.}
    \label{fig:blk-runtime}
  \end{subfigure}
  \hfill
  \begin{subfigure}[t]{0.32\textwidth}
    \centering
    \includegraphics[width=\linewidth]{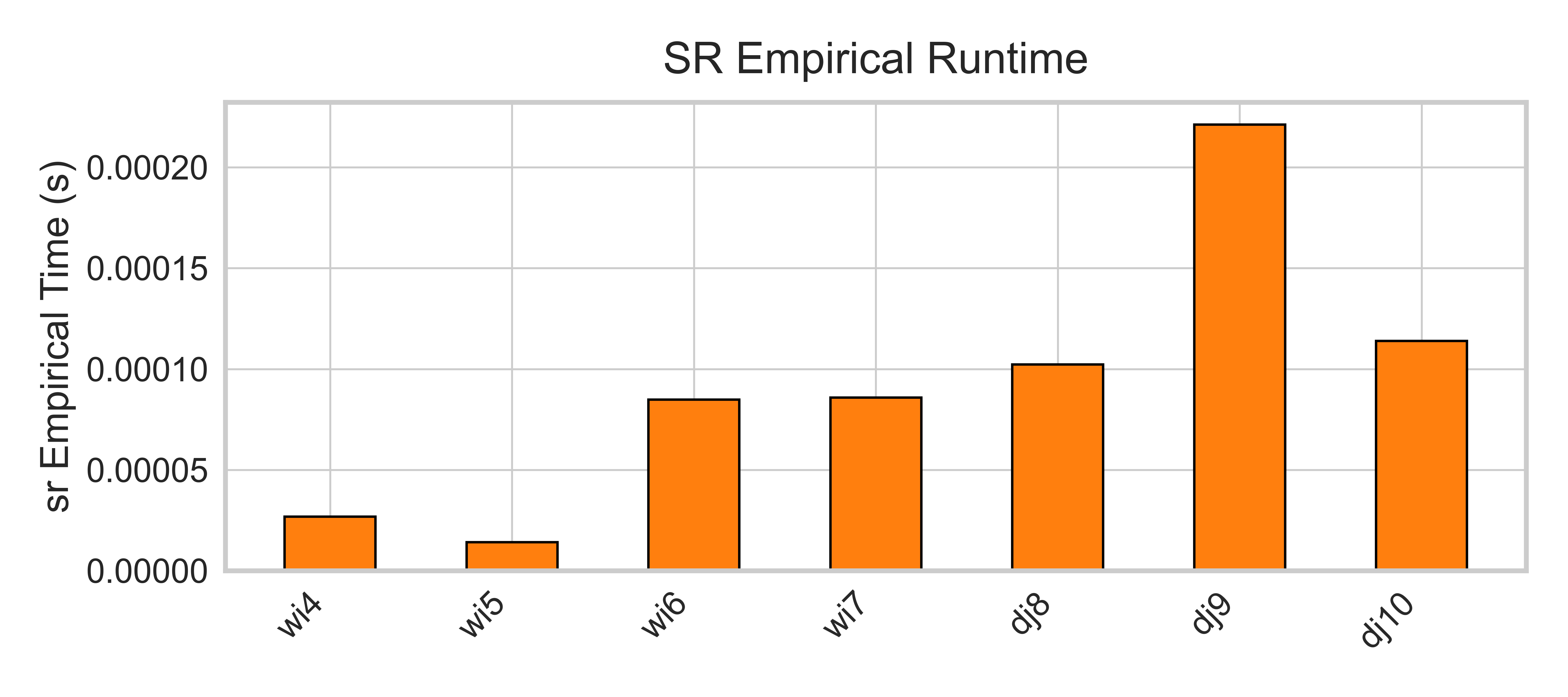}
    \caption{Empirical runtime after optional top-\(L\) compression.}
    \label{fig:topL-runtime}
  \end{subfigure}
  \caption{Empirical CPU runtimes for the three post-processing stages.
    The classical heavy-hitter stage has the largest footprint, the
    block-prefiltering stage is substantially cheaper, and the optional
    top-\(L\) compression is cheapest of all.}
  \label{fig:empirical-runtimes}
\end{figure}

% \begin{figure}[htbp]
%   \centering
%   \begin{subfigure}[t]{0.32\textwidth}
%     \centering
%     \includegraphics[width=\linewidth]{figs/hh_theory_proxy.png}
%     \caption{Proxy \( |\mathcal C_{\mathrm{HH}}|\,n^2 \) for classical heavy-hitter filtering.}
%     \label{fig:hh-theory}
%   \end{subfigure}
%   \hfill
%   \begin{subfigure}[t]{0.32\textwidth}
%     \centering
%     \includegraphics[width=\linewidth]{figs/blk_theory_proxy.png}
%     \caption{Proxy \( |\mathcal C_{\mathrm{blk}}|\,n^2 \) after block-prefiltering.}
%     \label{fig:blk-theory}
%   \end{subfigure}
%   \hfill
%   \begin{subfigure}[t]{0.32\textwidth}
%     \centering
%     \includegraphics[width=\linewidth]{figs/sr_theory_proxy.png}
%     \caption{Proxy \( |\mathcal C_{L}|\,n^2 \) after top-\(L\) compression.}
%     \label{fig:topL-theory}
%   \end{subfigure}
%   \caption{Simple operation-count proxies for the three post-processing
%     stages to visualize the candidate-set    reductions observed empirically.  The main theoretical comparison remains
%     the reduction from \(O(n^{k+3})\) classical post-processing under generic
%     thresholding to \(O(n^{k+2})\) under the quantum-informed threshold, while
%     the further top-\(L\) reduction should be interpreted as heuristic.}
%   \label{fig:theory-proxies}
% \end{figure}

\medskip

The empirical results illustrate that although the generic classical heavy-hitter filter already reduces the raw histogram, the quantum-informed threshold sharpens this reduction by one power of
\(n\), lowering the worst-case post-processing cost.  In addition, the block-prefiltering step
often yields a much stronger reduction in practice than is visible from the
worst-case bound alone. The optional top-\(L\) compression can shrink the
candidate set further, but it remains a heuristic empirical
accelerator. Crucially, we observe that the empirical quantum-informed refinement inherits the same
device-dependent noise sensitivity as the optimal-mass guarantee from
which its threshold is derived. We document this in Appendix~\ref{app:hh-noise-sensitivity}.
% Although the theorem-level guarantees in this section are stated for explicit histogram-based thresholding, in practice, one may implement the counting stage
% in several ways: with a direct hash table, with a streaming dictionary over
% observed outcomes, or with approximate counting data structures when memory is
% limited. These engineering choices can improve memory usage and wall-clock
% performance, but they are not needed for the formal heavy-hitter guarantees and
% are therefore kept separate from the main theoretical statements. Classical heavy-hitter thresholding yields a retained set of size
% \(O(n^{k+1})\), while the quantum-informed threshold reduces this to
% \(O(n^k)\). The block-prefiltering step often lowers the empirical footprint
% further. Thus the main rigorous gain of the heavy-hitter framework is already
% substantial: the CE--QAOA structure lowers the classical post-processing burden
% by one power of \(n\), and practical filtering can reduce it even more on
% structured instances.

\section{Conclusion and Outlook}
\label{sec:conclusion}

We developed a complete hybrid optimization layer around the CE--QAOA
sampler, built from structured OFM-kernel elements, and argued that
constrained quantum optimization should be viewed as an end-to-end
quantum--classical pipeline in which problem structure is exploited
throughout both the quantum and classical stages. We subsequently
introduced the NP-HQ algorithm as a concrete realization of this
paradigm. On the quantum side, the CE--QAOA sampler concentrates
probability on low-energy computational-basis states of a diagonal
problem Hamiltonian within an encoded constraint-preserving sector. It
uses ancilla-free one-hot initialization together with a two-local
block-\(XY\) mixer. On the classical side, deterministic checking,
repair, and scoring map noisy measurement outcomes to feasible
candidate solutions and evaluate them in polynomial time.

The performance guarantees we derived rest on the level of bias
achievable by the quantum sampler. In the instance-dependent regime
where the ideal CE--QAOA distribution assigns inverse-polynomial mass
to the optimal feasible set, we showed that the complete pipeline is an
exact-hit FPRAS\textsubscript{$q$}. Our total-variation analysis yields
a non-asymptotic guarantee and shows that this
FPRAS\textsubscript{$q$} guarantee persists within a device-noise window
determined by the ideal optimal mass and the instance-dependent
effective circuit depth. Concretely, we showed that the recovered
optimal mass remains inverse-polynomial, thereby yielding a polynomial
shot budget and polynomial classical overhead. Outside the exact-hit regime (whether because the ideal sampler does
not assign inverse-polynomial mass to the optimal feasible set or
because that mass is no longer guaranteed after noise and repair), our
algorithm continues to guarantee feasibility and yields an
instance-dependent \((1+\varepsilon)\)-approximation scheme whenever the
objective inflation induced by repair is controlled.

The complexity analysis also identifies a conditional
quantum--classical separation. Even if a classical baseline is granted
the same polynomial-time classical projection and scoring routine as
NP-HQ, and may sample directly from row-one-hot or feasible spaces,
inverse-polynomial overlap with the unknown global optimum uniformly
over an NP-hard kernel-admissible promise family would yield a
\(\mathrm{BPP}\) algorithm by standard amplification. Thus, unless
\(\mathrm{NP}\subseteq\mathrm{BPP}\), such a classical baseline cannot
uniformly reproduce the inverse-polynomial optimal-overlap guarantee
established here. The distinction, we argue, cannot be attributed merely to
constraint handling or repair; it lies in the structured sampling
distribution generated by the quantum kernel.

The heavy-hitter analysis sharpens this picture further. Once the
quantum kernel places sufficient weight on a small set of important
candidates, the classical post-processor need not inspect the full raw
output histogram. Thresholding the empirical histogram and checking
only the retained candidates can dramatically reduce memory and
runtime while preserving the exact-hit or approximation guarantees
available in the relevant regime. As documented in Appendix~\ref{app:hh-noise-sensitivity}, the empirical quantum-informed refinement inherits this device dependence. Nevertheless, we identify heavy-hitter
filtering as a natural interface between constrained quantum sampling
and scalable classical post-processing.

Finally, the hardware experiments instantiate the complete pipeline on
127-qubit IBM Eagle-r3 processors using \textsc{QOptlib}
travelling-salesman instances with up to 100 logical variables. Across
all tested instances, the final outputs match or improve the published
reference tours, providing a current-hardware demonstration of the
checker--repair--scoring architecture analyzed in this work.

% The framework also opens several directions for further development.
% The same checker--repair--scoring architecture can be specialized to
% other constrained encodings whenever the feasible structure admits an
% efficient deterministic projection. Further gains may come from
% device-aware compilation, tighter instance-specific repair bounds, and
% additional geometric information that permits more aggressive
% heavy-hitter compression. 

\appendix

%============================================================
%  MAIN-TEXT SUBSECTION
%============================================================

%============================================================
%============================================================
%  APPENDIX
%============================================================
%============================================================

\section{Mechanism behind the inherited Fej\'er success bound}
\label{app:fejer-mechanism}

\subsection{Effective Depth of Constraint–Enhanced QAOA Circuits}
\label{sec:ce-qaoa}
The Constraint--Enhanced QAOA (CE--QAOA) introduced in
Ref.~\cite{onahce} makes problem--algorithm co-design explicit through
the so-called OFM kernel, a designed product of blocks of one-hot or \(k\)-hot manifolds, and fixes the mixer family and initial state to match the
encoding. It further expands the co-design choice to include symmetries
derived from the hard cardinality constraints satisfied by valid
solutions. These facts lead to simple gate vs depth accounting.

% For an instance in the kernel with $m$ blocks of local dimension $n$, the initial state is
% \[
%   \ket{s_0}\;=\;\ket{s_{\mathrm{blk}}}^{\otimes m},
% \]
% and the mixer unitary is
% \[
%   U_M(\beta)\;=\;\bigotimes_{b=1}^{m}\exp\!\bigl(-i\beta\,\widetilde H_{XY}^{(b)}\bigr).
% \]
% A depth-$p$ CE--QAOA stack is
% \[
%   \ket{\psi_p(\vec\gamma,\vec\beta)}
%   \;=\;
%   \Bigl(\prod_{\ell=1}^{p} U_M(\beta_\ell)\,e^{-i\gamma_\ell H_C}\Bigr)\ket{s_0},
%   \qquad
%   \vec\gamma=(\gamma_1,\dots,\gamma_p),\;
%   \vec\beta=(\beta_1,\dots,\beta_p).
% \]

\begin{remark}[Gate vs depth]
\label{rem:gatedepth}
Usually, the problem Hamiltonian
contains \(\Theta(n^3)\) two-qubit $ZZ$  interaction terms per CE-QAOA
layer (the dense TSP construction is a good example) which can be parallelized to $O(n)$ sequential schedules. The complete \(XY\) mixer applies one exchange unitary to each pair of qubits within every one-hot block.  Since there are \(n\) blocks, each containing \(n\) qubits, the number of exchange operations
per mixing layer is
\[
N_{\mathrm{ex}}(n)
=
n\binom{n}{2}
=
\Theta(n^{3}).
\]
Exchanges acting on disjoint qubit pairs can be
executed in parallel, so an edge-colouring schedule gives
\(\Theta(n)\) sequential exchange rounds per mixing layer so the effective sequential depth remains
\(
G(n,p)
=
O(pn).
\) Initial state preparation requires depth of  $O(n)$ \cite{onahce}.
\end{remark}

\subsection{Finite Shot Finite Depth Guarantees}

This appendix provides a self-contained account of the analytic mechanism
underlying the dimension-free success bound inherited from Ref.~\cite{onahfinite}. The discussion is given in the
notation of the present paper and is organized as follows: Sec.~\ref{app:dephased} introduces the dephased
reference model that exposes the positive trigonometric filter; Sec.~\ref{app:factorization} derives the factorized reference distribution; and Sec.~\ref{app:consequences} draws out the finite-shot consequences.

%------------------------------------------------------------
\subsection{Dephased reference model}
\label{app:dephased}
%------------------------------------------------------------

The analytic strategy is to introduce an intermediate \emph{dephased}
model that retains the positive structure needed for lower bounding while
remaining analytically tractable. Define the cost-basis dephasing channel
\begin{equation}\label{eq:app-dephasing}
  \mathcal{T}(\rho)
  \;:=\;
  \int_0^{2\pi}\frac{d\phi}{2\pi}\;
    e^{-i\phi H_C}\,\rho\,e^{+i\phi H_C}.
\end{equation}
This channel projects onto the diagonal in the $H_C$ eigenbasis: it
removes all coherences between distinct eigenspaces and leaves only the
occupation probabilities $\bra{z}\rho\ket{z}$.

\begin{remark}[Role of dephasing]
The dephasing channel is \emph{not} a physical noise model applied during
the computation. It is an analytic device: by inserting $\mathcal{T}$
after each CE--QAOA layer, we obtain a reference model that is easier to
bound from below.
\end{remark}

After inserting $\mathcal{T}$ between each pair of layers, the evolution
of the diagonal is governed by a classical Markov chain. Write
$v^{(r)}(z) := \bra{z}\rho_r\ket{z}$ for the diagonal distribution
after $r$ dephased layers. The mixer step acts as
\begin{equation}\label{eq:app-markov}
  v^{(r)}(z)
  \;=\;
  \sum_{y \in  X} M_{\beta_r}(z|y)\;v^{(r-1)}(y),
  \qquad
  M_\beta(z|y)
  \;:=\;
  \bigl|\!\bra{z}U_M(\beta)\ket{y}\!\bigr|^2.
\end{equation}
Each transition matrix $M_\beta$ is \emph{unistochastic} (derived from
the squared moduli of a unitary matrix) and hence doubly stochastic.

\begin{lemma}[Properties of the mixer transition matrix]
\label{lem:app-mixer-properties}
For every $\beta \in \mathbb{R}$:
\begin{enumerate}[label=\textup{(\alph*)},leftmargin=*]
  \item $M_\beta$ is doubly stochastic on $\OH$:
    $\sum_z M_\beta(z|y) = 1$ and $\sum_y M_\beta(z|y) = 1$.
  \item $M_\beta$ preserves the uniform distribution:
    if $v(z) = 1/D$ for all $z$, then
    $[M_\beta\,v](z) = 1/D$.
  \item At $\beta = 0$, $M_0 = I$ (the identity permutation on
    distributions).
  \item Because $U_M(\beta)$ factorizes over blocks, $M_\beta$
    factorizes as a tensor product of $m$ local transition matrices,
    each of size $n_{\mathrm{loc}} \times n_{\mathrm{loc}}$.
\end{enumerate}
\end{lemma}

\begin{proof}
(a) follows from unitarity of $U_M(\beta)$. (b) is immediate from (a).
(c) is immediate from $U_M(0) = I$. (d) follows from the tensor-product
structure of $U_M(\beta)$.
\end{proof}

Starting from the initial diagonal
$v^{(0)}(z) = |\!\braket{z}{s_0}\!|^2 = D^{-1}$
(uniform on $\OH$), iterating the Markov chain gives the \emph{mixer
envelope}:
\begin{equation}\label{eq:app-Wp}
  W_p(z;\boldsymbol\beta)
  \;:=\;
  \bigl[M_{\beta_p}\cdots M_{\beta_1}\,v^{(0)}\bigr](z).
\end{equation}
Because each $M_{\beta_r}$ is doubly stochastic and $v^{(0)}$ is
uniform, $W_p$ is a probability distribution on $\OH$. Its role in the
success bound is to quantify how the mixer schedule redistributes mass
across the encoded manifold, independently of the cost structure.

\begin{lemma}[Mixer envelope is a probability distribution]
\label{lem:app-envelope-prob}
For every angle schedule $\boldsymbol\beta = (\beta_1,\dots,\beta_p)$,
$W_p(\cdot;\boldsymbol\beta) \ge 0$ and
$\sum_{z \in \OH} W_p(z;\boldsymbol\beta) = 1$.
\end{lemma}

\begin{proof}
Nonnegativity follows from $M_\beta(z|y) \ge 0$. Normalization follows
by induction: $v^{(0)}$ is normalized, and each doubly stochastic step
preserves the total mass.
\end{proof}

The \emph{Fej\'er kernel} of order $p$ is the nonnegative trigonometric
polynomial:
\begin{equation}\label{eq:app-fejera}
  F_p(\vartheta)
  \;:=\;
  \frac{1}{p+1}
  \left|
    \sum_{r=0}^{p} e^{ir\vartheta}
  \right|^2
  \;=\;
  \frac{1}{p+1}
  \left(
    \frac{\sin\!\bigl(\frac{(p+1)\vartheta}{2}\bigr)}
         {\sin\!\bigl(\frac{\vartheta}{2}\bigr)}
  \right)^{\!2}.
\end{equation}

For all $|\vartheta| \ge \delta > 0$ \cite{onahfinite},
    \begin{equation}\label{eq:app-Mp-bound}
      F_p(\vartheta)
      \;\le\;
      \frac{1}{(p+1)\sin^2(\delta/2)}.
    \end{equation}

The key feature is the contrast between the on-peak value
$F_p(0) = p + 1$ and the off-peak bound
$M_p(\delta) = O\!\bigl((p+1)^{-1}\bigr)$: the ratio is
$(p+1)^2\sin^2(\delta/2)$, which grows quadratically with filter depth.
This is the mechanism through which deeper CE--QAOA alternations can enhance
discrimination between optimal and suboptimal cost phases.

%------------------------------------------------------------
\subsection{Factorized reference distribution}
\label{app:factorization}
%------------------------------------------------------------

The dephased reference model produces a distribution on $\OH$ that
factorizes into two components: the mixer envelope $W_p$ and the Fej\'er
filter $F_p$\cite{onahfinite}. To see this, consider the effect of $p$ dephased layers on
the initial state. After the \(r\)-th cost layer, the off-diagonal coherences acquire
cost-dependent phases. Subsequent mixing and dephasing transfer this
phase dependence into the diagonal reference weights. After dephasing, these phases
manifest as a multiplicative weight on each basis string. The dephased dynamics determines the mixer envelope
\(W_p(z;\boldsymbol\beta)\). Under the harmonic phase schedule
\(\gamma_r=r\gamma\), the Dirichlet-filtered reference construction of
Ref.~\cite{onahfinite} weights this envelope by the Fej\'er factor
\(F_p(\theta(z)-\theta^\star)\). The resulting diagonal weight is
therefore proportional to
\(W_p(z;\boldsymbol\beta)F_p(\theta(z)-\theta^\star)\). 

The normalized \emph{factorized reference distribution} is:
\begin{equation}\label{eq:app-factorizationa}
  \Pr_p^{\mathrm{ref}}[z]
  \;:=\;
  \frac{
    W_p(z;\boldsymbol\beta)\;
    F_p\!\bigl(\theta(z) - \theta^\star\bigr)
  }{
    \displaystyle
    \sum_{y \in X}
      W_p(y;\boldsymbol\beta)\;
      F_p\!\bigl(\theta(y) - \theta^\star\bigr)
  }.
\end{equation}

This factorization isolates the two competing effects:
\begin{itemize}[leftmargin=*]
  \item The envelope $W_p$ controls how the mixer spreads mass across the
    encoded manifold. A string that receives little mixer mass will
    contribute little to the success probability regardless of its cost
    phase.
  \item The Fej\'er factor $F_p(\theta(z) - \theta^\star)$ amplifies
    strings near the optimal phase and suppresses those far from it. The
    degree of amplification grows with the filter order $p$.
\end{itemize}

Define the \emph{optimal-set envelope weight}:
\(
  C_\beta
  \;:=\;
  \sum_{x \in \Omega^\star} W_p(x;\boldsymbol\beta).
\)

\begin{theorem}[Reference Fej\'er factorization and success bound]
\label{thm:app-fejer-success}
Assume that the TSP specialization admits the factorized reference law \eqref{eq:app-factorization}, and assume a positive wrapped phase
separation $\delta > 0$ in the sense of \eqref{eq:app-phase-gap}. Then
the probability of sampling an optimal basis string satisfies
\begin{equation}\label{eq:app-q0}
  q_0
  \;:=\;
  \Pr_p^{\mathrm{ref}}[\Omega^\star]
  \;\ge\;
  \frac{(p+1)\,C_\beta}
       {(p+1)\,C_\beta \;+\; M_p(\delta)\,(1 - C_\beta)}.
\end{equation}
In particular, the lower bound is \emph{dimension-free}: it depends only
on the filter order $p$, the phase gap $\delta$, and the optimal-set
envelope weight $C_\beta$.
\end{theorem}

\begin{proof}
Summing \eqref{eq:app-factorization} over $x \in \Omega^\star$ and using
$F_p(0) = p + 1$ gives
\[
  q_0
  \;=\;
  \frac{(p+1)\,C_\beta}
       {\displaystyle
        \sum_{y \in  X}
          W_p(y;\boldsymbol\beta)\,
          F_p\!\bigl(\theta(y) - \theta^\star\bigr)
       }.
\]
Split the denominator into optimal and suboptimal contributions:
\[
  \text{denom}
  \;=\;
  (p+1)\,C_\beta
  \;+\;
  \sum_{y \notin \Omega^\star}
    W_p(y;\boldsymbol\beta)\,
    F_p\!\bigl(\theta(y) - \theta^\star\bigr).
\]
For every $y \notin \Omega^\star$, the phase-gap assumption implies
$|\theta(y) - \theta^\star| \ge \delta$ on $\mathbb{T}$, so
$F_p\!\bigl(\theta(y) - \theta^\star\bigr) \le M_p(\delta)$ by
Eq. \ref{eq:app-Mp-bound}. Hence the suboptimal sum is at most
\[
  M_p(\delta)
  \sum_{y \notin \Omega^\star} W_p(y;\boldsymbol\beta)
  \;=\;
  M_p(\delta)\,(1 - C_\beta),
\]
which yields \eqref{eq:app-q0}. See Ref. \cite{onahfinite} for details
\end{proof}

% \begin{proof}
% Rearranging the requirement
% $q_0 \ge 1 - \varepsilon$ using \eqref{eq:app-q0} and the bound
% $M_p(\delta) \le [(p+1)\sin^2(\delta/2)]^{-1}$ yields
% \eqref{eq:app-depth-condition}.
% \end{proof}

%------------------------------------------------------------
\subsection{Consequences for the hybrid pipeline}
\label{app:consequences}
%------------------------------------------------------------

We now translate the reference-model bound into the guarantees used in the main text.

\begin{corollary}[Instance-dependent inverse-polynomial mass from the
joint envelope--filter law]
\label{cor:app-inverse-poly}
Consider a problem family admitting the factorized reference law
\eqref{eq:app-factorization}. Suppose that, for each size-\(n\) instance
and the chosen mixer schedule, the corresponding envelope and phase
parameters satisfy
\[
C_\beta \geq c_0 n^{-a},
\qquad
\sin^2(\delta/2)\geq c_1,
\]
for constants \(c_0,c_1>0\) and \(a\geq0\) independent of \(n\).
Then, for any fixed depth \(p\), the joint envelope--filter law gives
\[
q_0
\geq
\frac{(p+1)c_0n^{-a}}
     {(p+1)c_0n^{-a}
      +\bigl[(p+1)c_1\bigr]^{-1}}
=
\Omega(n^{-a}).
\]
Setting \(k:=a\), this supplies the inverse-polynomial premise
\[
q_0=\Omega(n^{-k})
\]
used in Theorem~\ref{thm:noise_robust} and the
\(\mathrm{FPRAS}_q\) analysis.
\end{corollary}

\begin{proof}
By Theorem~\ref{thm:app-fejer-success},
\[
q_0
\geq
\frac{(p+1)C_\beta}
     {(p+1)C_\beta+M_p(\delta)(1-C_\beta)}.
\]
The right-hand side is increasing in \(C_\beta\) and decreasing in
\(M_p(\delta)\). Moreover,
\[
M_p(\delta)
\leq
\frac{1}{(p+1)\sin^2(\delta/2)}
\leq
\frac{1}{(p+1)c_1}.
\]
Using \(C_\beta\geq c_0n^{-a}\) and \(1-C_\beta\leq1\) therefore gives
\[
q_0
\geq
\frac{(p+1)c_0n^{-a}}
     {(p+1)c_0n^{-a}
      +\bigl[(p+1)c_1\bigr]^{-1}}.
\]
For fixed \(p\), the numerator is
\(\Theta(n^{-a})\), while the denominator is
\(\Theta(1)\). Hence \(q_0=\Omega(n^{-a})\).
\end{proof}

\begin{corollary}[Sufficient depth for target success probability\cite{onahfinite}]
\label{cor:app-depth}
To achieve $q_0 \ge 1 - \varepsilon$ in the reference model, it suffices
to choose the filter order $p$ such that
\begin{equation}\label{eq:app-depth-condition}
  (p+1)^2
  \;\ge\;
  \frac{1 - \varepsilon}{\varepsilon}
  \cdot
  \frac{1 - C_\beta}{C_\beta}
  \cdot
  \csc^2\!(\delta/2).
\end{equation}
In particular, if $C_\beta = \Omega(n^{-a})$ and
$\sin^2(\delta/2) = \Omega(1)$, then $p = O(n^{a/2})$ suffices for any
fixed $\varepsilon > 0$.
\end{corollary}

\subsection{Lipschitz Continuity of Noisy Output Distributions}
\label{sec:lipschitz-noise}

% By standard continuity bounds for quantum channels, if each two‐qubit gate deviates by at most $\epsilon$ in the diamond norm from its ideal implementation, then the overall $G(n)$‐gate circuit deviates by at most $G(n)\,\epsilon$ in diamond norm \cite{Nielsen2000, Wilde2017}.  Moreover, the classical total‐variation distance between output distributions is bounded by the trace distance of the underlying states.  Combining these facts (and noting $G(n)=O(pn)$) yields the Lipschitz bound
% \[
%   \bigl\|P_{\epsilon}-P_{0}\bigr\|_{\mathrm{TV}}
%   \;\le\;
%   \frac12G(n)\,\epsilon,
% \]
% as used in Theorem~\ref{thm:noise_robust}.  An equivalent composition argument appears in \cite[Ch.~4 \& 9]{Wilde2017}. %\cite[Sec.~9.2]{Nielsen2000}

Under the effective-layer noise model of Sec.~\ref{sec:complexity}, the measured
output distribution obeys the following Lipschitz bound. For
probability distributions \(P\) and \(Q\) on the same finite sample
space \(\Omega\), recall that
\[
  \|P-Q\|_{\mathrm{TV}}
  =
  \frac12\sum_{x\in\Omega}|P(x)-Q(x)|
  =
  \max_{A\subseteq\Omega}|P(A)-Q(A)|.
\]

\begin{lemma}[Lipschitz continuity under local noise]
\label{lem:lipschitz-noise}
Suppose that the noisy implementation \(\mathcal E_j\) of each of the
\(G(n)\) sequential noisy locations satisfies
\[
  \|\mathcal E_j-\mathcal U_j\|_\diamond
  \le C\epsilon
\]
for some constant \(C>0\), where \(\mathcal U_j\) is the corresponding
ideal channel. Then
\[
  \|P_\epsilon-P_0\|_{\mathrm{TV}}
  \le
  \frac{C}{2}G(n)\epsilon
  =:
  L G(n)\epsilon .
\]
In particular, when \(C=1\),
\[
  \|P_\epsilon-P_0\|_{\mathrm{TV}}
  \le
  \frac12G(n)\epsilon .
\]
\end{lemma}

\begin{proof}
A telescoping expansion of the noisy and ideal channel compositions,
together with the triangle inequality and contractivity of the diamond
norm under composition with quantum channels, gives
\[
  \|\mathcal E_{G}\circ\cdots\circ\mathcal E_{1}
    -
    \mathcal U_{G}\circ\cdots\circ\mathcal U_{1}\|_\diamond
  \le
  \sum_{j=1}^{G(n)}
  \|\mathcal E_j-\mathcal U_j\|_\diamond
  \le
  C G(n)\epsilon .
\]
Consequently,
\(
  \|\rho_\epsilon-\rho_0\|_1
  \le C G(n)\epsilon .
\) Let \(\mathcal M\) denote computational-basis measurement. By
contractivity of trace distance under CPTP maps,
\[
  \begin{aligned}
  \|P_\epsilon-P_0\|_{\mathrm{TV}}
  &=
  \frac12
  \|\mathcal M(\rho_\epsilon)-\mathcal M(\rho_0)\|_1 
  \le
  \frac12\|\rho_\epsilon-\rho_0\|_1
  \le
  \frac{C}{2}G(n)\epsilon .
  \end{aligned}
\]
See Refs.~\cite{Nielsen2000,Wilde2017} for the standard channel-norm
and trace-distance properties used above.
\end{proof}

\subsection{Bound on Success Probability in Presence of Noise}

Let $\mathcal{X}^*\subseteq\Omega$ denote the set of optimal strings, and define
\[
  p_0 = P_0(\mathcal{X}^*),
  \quad
  p_{\min}(\epsilon) = P_{\epsilon}(\mathcal{X}^*).
\]
By the defining event bound for total variation,
\[
  \bigl|p_{\min}(\epsilon)-p_0\bigr|
  =
  \bigl|P_\epsilon(X^\star)-P_0(X^\star)\bigr|
  \le
  \|P_\epsilon-P_0\|_{\mathrm{TV}}
  \le
  L G(n,p)\epsilon .
\]
Hence,
\(
  p_{\min}(\epsilon)
  \ge
  p_0-LG(n,p)\epsilon .
\)

Treating each sample as an independent Bernoulli trial with success probability at least $p_{\min}(\epsilon)$, the probability of not seeing an optimal string in $S$ shots is
\[
  (1 - p_{\min}(\epsilon))^S \le \exp\bigl(-S\,p_{\min}(\epsilon)\bigr).
\]

To ensure an overall failure probability at most \(\delta\), it is
enough to choose
\[
S
\ge
\left\lceil
\frac{\ln(1/\delta)}
     {p_0-LG(n,p)\epsilon}
\right\rceil ,
\]
provided \(p_0-LG(n,p)\epsilon>0\). A sufficient condition for
\(p_{\min}(\epsilon)\ge n^{-k-1}\) is

\[
\epsilon
\le
\frac{p_0-n^{-k-1}}{LG(n,p)}
=:\epsilon_{\mathrm{th}}.
\]

Below this threshold, the solver succeeds exactly with $\mathrm{poly}(n)$ shots.

% Heavy‑Hitter QAOA: Depth‑Compressed Circuits Achieving Quantum Advantage for Constrained Optimisation

\subsection{ Heavy Hitters From Depth‑Compressed Noisy Quantum  Circuits }
\label{sec:depth-planning}

\begin{theorem}[Combined TV-distance bound on circuit approximation]
\label{thm:combined_tv}
Let
\[
U_{\rm PC},
\qquad
U_{\rm APC},
\]
be the ideal unitaries of the original circuit \textbf{PC} and its
compressed approximation \textbf{APC}, respectively, acting on \(n\)
qubits. Define their ideal output distributions by
\[
P_{\rm PC}(x)
\equiv
P_0(x)
:=
\bigl|\langle x|U_{\rm PC}|0^n\rangle\bigr|^2,
\qquad
P_{\rm APC}(x)
:=
\bigl|\langle x|U_{\rm APC}|0^n\rangle\bigr|^2,
\]
for \(x\in\{0,1\}^n\).

Suppose the circuit approximation satisfies
\[
\bigl\|U_{\rm PC}-U_{\rm APC}\bigr\|_{\rm tr}
\leq \delta,
\]
and consequently
\[
\bigl\|P_{\rm APC}-P_0\bigr\|_{\rm TV}
\leq \delta.
\]

Now consider hardware noise of strength \(\epsilon\). Let
\(P_{\epsilon,\delta}\) denote the noisy output distribution obtained
by running \textbf{APC}, and let \(G_{\rm APC}(n)\) denote its effective
noisy depth. Applying the Lipschitz bound of
Appendix~\ref{sec:lipschitz-noise} to the approximate circuit gives
\[
\bigl\|P_{\epsilon,\delta}-P_{\rm APC}\bigr\|_{\rm TV}
\leq
L\,G_{\rm APC}(n)\,\epsilon,
\qquad
L>0,
\quad
G_{\rm APC}(n)\geq 1.
\]

Then
\[
\bigl\|P_{\epsilon,\delta}-P_0\bigr\|_{\rm TV}
\leq
\underbrace{L\,G_{\rm APC}(n)\,\epsilon}_{\text{noise error}}
+
\underbrace{\delta}_{\text{approximation error}}.
\]

Let \(x^\star\) be an optimal string and suppose its ideal probability
satisfies
\[
p_0^\star
:=
P_0(x^\star)
\geq
\frac{c}{n^k}
\]
for fixed constants \(c>0\) and \(k\geq 0\). Its probability under the
noisy approximate circuit then obeys
\[
p_{\epsilon,\delta}^\star
:=
P_{\epsilon,\delta}(x^\star)
\geq
p_0^\star
-
L\,G_{\rm APC}(n)\,\epsilon
-
\delta.
\]

For the heavy-hitter threshold \(1/H\), define
\[
\mathcal H_{1/H}
:=
\left\{
x:
P_{\epsilon,\delta}(x)\geq\frac{1}{H}
\right\}.
\]
If
\[
p_0^\star
-
L\,G_{\rm APC}(n)\,\epsilon
-
\delta
\geq
\frac{1}{H},
\]
then
\[
x^\star\in\mathcal H_{1/H},
\qquad
\left|\mathcal H_{1/H}\right|\leq H.
\]
\end{theorem}

\begin{proof}[Sketch]
\textnormal{(a)}
By the triangle inequality for total variation,
\[
\begin{aligned}
\bigl\|P_{\epsilon,\delta}-P_0\bigr\|_{\rm TV}
&\leq
\bigl\|P_{\epsilon,\delta}-P_{\rm APC}\bigr\|_{\rm TV}
+
\bigl\|P_{\rm APC}-P_0\bigr\|_{\rm TV}
\\
&\leq
L\,G_{\rm APC}(n)\,\epsilon+\delta.
\end{aligned}
\]

\textnormal{(b)}
Applying the defining event bound for total variation to the singleton
event \(\{x^\star\}\) gives
\[
\begin{aligned}
p_{\epsilon,\delta}^\star
&=
P_{\epsilon,\delta}(x^\star)
\\
&\geq
P_0(x^\star)
-
\bigl\|P_{\epsilon,\delta}-P_0\bigr\|_{\rm TV}
\\
&\geq
p_0^\star
-
L\,G_{\rm APC}(n)\,\epsilon
-
\delta.
\end{aligned}
\]

\textnormal{(c)}
If
\[
p_0^\star
-
L\,G_{\rm APC}(n)\,\epsilon
-
\delta
\geq
\frac{1}{H},
\]
then \(P_{\epsilon,\delta}(x^\star)\geq 1/H\), and hence
\(x^\star\in\mathcal H_{1/H}\).

Finally, normalization gives
\[
1
\geq
\sum_{x\in\mathcal H_{1/H}}
P_{\epsilon,\delta}(x)
\geq
\frac{\left|\mathcal H_{1/H}\right|}{H},
\]
so
\[
\left|\mathcal H_{1/H}\right|\leq H.
\]
\end{proof}

\begin{corollary}[Threshold-vs-Depth Trade-off]
\label{thm:depth_vs_threshold} 
Let $D$ be the original two-qubit–layer depth and  $\epsilon_{L}$ be the per-layer error rate. Let $\alpha\in(0,1]$ be the depth‐reduction factor (so the approximate circuit has depth $\alpha D$), and $\delta$ satisfy 
    \(
      \bigl\|P_{\rm APC}-P_{\rm PC}\bigr\|_{\rm TV}\;\le\;\delta.
    \)
Then under the same Lipschitz‐in‐TV noise bound above in Theorem \ref{thm:combined_tv}
$\|P_{\epsilon}-P_{0}\|_{\rm TV}\le D\,\epsilon_{L}$, the new heavy‐hitter threshold
\[
  \theta_{\rm new}
  \;\le\;
  p_{0}
  \;-\;
  \bigl(\alpha D\,\epsilon_{L}\bigr)
  \;-\;
  \delta
\]
guarantees the optimum remains in $\{x:P_{\epsilon,\rm apx}(x)\ge\theta_{\rm new}\}$.
Moreover, comparing to the old threshold
$\theta_{\rm old}=p_{0}-D\,\epsilon_{L}$, the shift
\(
  \Delta\theta
  =\theta_{\rm new}-\theta_{\rm old}
  =(1-\alpha)\,D\,\epsilon_{L}-\delta
\)
implies \emph{you only need to lower your cutoff} iff
\(
  \delta \;>\;(1-\alpha)\,D\,\epsilon_{L}.
\)
\end{corollary}

\begin{algorithm}[H]
\caption{Depth-Compressed Heavy-Hitter Threshold Adjustment}
\label{alg:HH-threshold}
\begin{algorithmic}[1]

\Require Original two-qubit-layer depth \(D\), Per-layer error rate \(\epsilon_L\),  Depth-reduction factor \(\alpha\in[0,1]\), Approximation error \(\delta\in[0,1]\) in total variation distance and Ideal estimated success probability \(p_0\)

\Ensure New heavy-hitter cutoff \(\theta_{\mathrm{new}}\),
\(\mathrm{action}\in
\{\textsc{Raise},\textsc{Lower},\textsc{Unchanged}\}\), Adjustment magnitude \(\lvert\Delta\theta\rvert\)

\State
\(\theta_{\mathrm{old}}
    \gets p_0-D\epsilon_L\)
\Comment{Original cutoff}

\State
\(\theta_{\mathrm{new}}
    \gets p_0-\alpha D\epsilon_L-\delta\)
\Comment{Compressed cutoff}

\State
\(\Delta\theta
    \gets\theta_{\mathrm{new}}-\theta_{\mathrm{old}}\)

\If{\(\Delta\theta>0\)}
    \State \(\mathrm{action}\gets\textsc{Raise}\)
\ElsIf{\(\Delta\theta<0\)}
    \State \(\mathrm{action}\gets\textsc{Lower}\)
    \State \(\Delta\theta\gets-\Delta\theta\)
    \Comment{Take absolute value}
\Else
    \State \(\mathrm{action}\gets\textsc{Unchanged}\)
\EndIf

\State \Return
\(\bigl(\theta_{\mathrm{new}},
        \mathrm{action},
        \Delta\theta\bigr)\)

\end{algorithmic}
\end{algorithm}

We can visualize the effect of Algorithm \ref{alg:HH-threshold} on $\Delta\theta(\alpha,\delta)$ for fixed $D,\epsilon_{L},p_0$  from the derived relation 
\(
  \Delta\theta(\alpha,\delta)
  = (1-\alpha)\,D\,\epsilon_{L} - \delta
\) in Figure~\ref{fig:threshold_plot} where it shows a two‐dimensional contour map of the maximum permissible heavy‐hitter threshold
\[
  \theta_{\max}(\alpha,\delta)
  = p_{0} \;-\;\alpha\,D\,\epsilon_{L}\;-\;\delta,
\]
plotted over
\(
  0 \le \alpha \le 1,
  \quad
  0 \le \delta \le 0.30,
\)
for the illustrative parameters
\(\;p_{0}=0.2,\;D=100,\;\epsilon_{L}=4.55\times10^{-3}.\)

\begin{figure}[ht]
  \centering
  \begin{minipage}[t]{0.52\textwidth}
    \vspace{0pt}
    \centering
    \includegraphics[width=\linewidth]{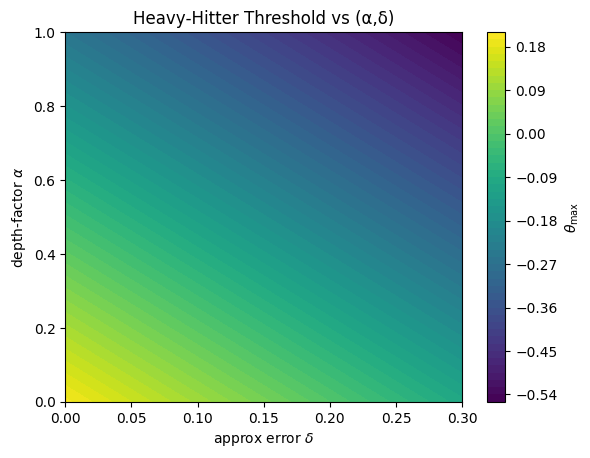}
  \end{minipage}
  \hfill
  \begin{minipage}[t]{0.44\textwidth}
    \vspace{0pt}
    \caption{\small
    Maximum heavy-hitter threshold
    \(\theta_{\max}(\alpha,\delta)\)     as a function of the depth factor \(\alpha\) and approximation error
    \(\delta\). At fixed \(\delta\), depth compression increases
    \(\theta_{\max}\) linearly; at fixed \(\alpha\), increasing \(\delta\)
    decreases it one-to-one. The bottom-left region shows that aggressive
    depth reduction and negligible approximation error permit a cutoff near
    \(0.18\). Moving upward or rightward lowers the threshold through the
    combined error budget \(\alpha D\epsilon_L+\delta\). In the darkest
    region, this budget exceeds the ideal success probability \(p_0=0.2\),
    giving \(\theta_{\max}<0\).
    }
    \label{fig:threshold_plot}
  \end{minipage}
\end{figure}

  % \item \textbf{Horizontal axis:}  Moving right corresponds to allowing a larger discrepancy between the compressed circuit’s output distribution and the ideal one.

  % \item \textbf{Vertical axis:}  \(\alpha=1\) means no depth reduction (full original depth), while \(\alpha=0\) represents a hypothetically “zero‐depth” (i.e.\ completely trivial) circuit.

In practice, one chooses a desired heavy‐hitter cutoff \(\theta\) (e.g.\ based on the mean amplitude) and then reads off from this plot the allowable combinations of \((\alpha,\delta)\) that satisfy \(\theta\le\theta_{\max}(\alpha,\delta)\). Shot-Budget shifts visibly under the $(\alpha,\delta)$ trade-off. From the heavy–hitter probability bound proved above, \(
  p_{\min}(\alpha,\delta)
  \;\ge\;
  p_{0}-\alpha D\epsilon_{L}-\delta,\) 
feeds directly into the Chernoff–style shot formula
(Theorem~\ref{thm:phqc-exact}). %:

\subsection{Noise sensitivity of the quantum-informed heavy-hitter refinement}
\label{app:hh-noise-sensitivity}

The quantum-informed heavy-hitter refinement inherits the
device-dependent character of the optimal-mass guarantee on which its
threshold is based. In particular, the refined candidate set is
guaranteed to retain the optimum only while the noisy output
distribution remains inside the total-variation window assumed in
Theorem~\ref{thm:quantum-informed-HH}. Tables
\ref{tab:full-hh-k3} and \ref{tab:full-hh-k4} report the full empirical
data for \(k=3\) and \(k=4\), respectively. Here
\(\lvert H_{\mathrm{HH}}\rvert\) is the candidate count after ordinary
heavy-hitter filtering, \(\lvert H_{\mathrm{blk}}\rvert\) is the count
after block filtering, and \(\lvert H_{\mathrm{SR}}\rvert\) is the
count after the final structure-refined quantum filtering step. Each
reported gap is the percentage increase in the best recovered cost
relative to exhaustive analysis of the corresponding raw sample
histogram. Repeated rows for an instance correspond to independent
hardware runs or different shot budgets.

The ordinary heavy-hitter stage reproduces the exhaustive best cost in
almost every finite case. Additional structure-aware compression
usually preserves that cost while substantially reducing the candidate
set, but a minority of runs exhibit a nonzero gap. These departures
are consistent with the device-dependent scope of the quantum-informed
threshold. Once noise depletes or redistributes the mass supporting the
best candidate beyond the admissible window, aggressive filtering need
not retain it. This effect is more visible for \(k=3\), where several
runs produce nonzero gaps and some noisy histograms yield an empty
retained set. The corresponding entries are reported as
\(\infty\), because no finite recovered cost exists from which to form
a percentage gap. Increasing to \(k=4\) enlarges the retained sets and
eliminates empty-set failures in the reported data, although isolated
departures remain after the more aggressive refinement.

\begin{table*}[p]
\centering
\caption{Full heavy-hitter data for \(k=3\). The three gap columns give
the percentage increase in the best recovered cost relative to
exhaustive analysis of the raw sample histogram. An entry of
\(\infty\) indicates that the corresponding retained set was empty.}
\label{tab:full-hh-k3}

\scriptsize
\setlength{\tabcolsep}{3pt}
\renewcommand{\arraystretch}{0.90}

\resizebox{\textwidth}{!}{%
\begin{tabular}{lrrrrrrr}
\toprule
Instance & Shots
& \(\lvert H_{\mathrm{HH}}\rvert\)
& \(\lvert H_{\mathrm{blk}}\rvert\)
& \(\lvert H_{\mathrm{SR}}\rvert\)
& Gap HH [\%] & Gap blk [\%] & Gap SR [\%] \\
\midrule
dj10  & 50000  & 50000 & 1000 & 1000 & 0.00 & 6.25  & 6.25 \\
dj10  & 5000   & 5000  & 1000 & 1000 & 0.00 & 7.99  & 7.99 \\
dj10  & 4096   & 4096  & 1000 & 1000 & 0.00 & 4.89  & 4.89 \\
dj8   & 32000  & 32000 & 512  & 512  & 0.00 & 0.07  & 0.07 \\
dj8   & 32000  & 32000 & 512  & 512  & 0.00 & 0.00  & 0.00 \\
dj8   & 32000  & 32000 & 512  & 512  & 0.00 & 0.00  & 0.00 \\
dj8   & 4096   & 4096  & 512  & 512  & 0.00 & 0.00  & 0.00 \\
dj9   & 40500  & 40500 & 729  & 729  & 0.00 & 0.14  & 0.14 \\
dj9   & 40500  & 40500 & 729  & 729  & 0.00 & 0.61  & 0.61 \\
dj9   & 4050   & 4050  & 729  & 729  & 0.00 & 3.80  & 3.80 \\
dj9   & 4096   & 4096  & 729  & 729  & 0.00 & 4.97  & 4.97 \\
eil10 & 50000  & 50000 & 1000 & 1000 & 0.00 & 17.75 & 17.75 \\
eil10 & 5000   & 5000  & 1000 & 1000 & 0.00 & 0.00  & 0.00 \\
eil10 & 4096   & 4096  & 1000 & 1000 & 0.00 & 4.23  & 4.23 \\
eil11 & 60500  & 60500 & 1331 & 1331 & 0.00 & 16.17 & 16.17 \\
eil11 & 6050   & 6050  & 1331 & 1331 & 0.00 & 3.47  & 3.47 \\
eil11 & 4096   & 4096  & 1331 & 1331 & 0.00 & 0.00  & 0.00 \\
wi4   & 8000   & 1     & 1    & 1    & 0.00 & 0.00  & 0.00 \\
wi4   & 8000   & 7     & 7    & 7    & 6.70 & 6.70  & 6.70 \\
wi4   & 800    & 793   & 64   & 64   & 0.00 & 0.00  & 0.00 \\
wi4   & 10000  & 230   & 64   & 64   & 0.00 & 0.00  & 0.00 \\
wi4   & 100000 & 234   & 64   & 64   & 0.00 & 0.00  & 0.00 \\
wi4   & 100000 & 216   & 64   & 64   & 0.00 & 0.00  & 0.00 \\
wi4   & 4096   & 449   & 64   & 64   & 0.00 & 0.00  & 0.00 \\
wi5   & 12500  & 0     & 0    & 0    & \(\infty\) & \(\infty\) & \(\infty\) \\
wi5   & 12500  & 0     & 0    & 0    & \(\infty\) & \(\infty\) & \(\infty\) \\
wi5   & 1250   & 1250  & 125  & 125  & 0.00 & 0.00  & 0.00 \\
wi5   & 4096   & 4095  & 125  & 125  & 0.00 & 0.00  & 0.00 \\
wi6   & 18000  & 0     & 0    & 0    & \(\infty\) & \(\infty\) & \(\infty\) \\
wi6   & 18000  & 0     & 0    & 0    & \(\infty\) & \(\infty\) & \(\infty\) \\
wi6   & 1800   & 1800  & 216  & 216  & 0.00 & 0.00  & 0.00 \\
wi6   & 4096   & 4096  & 216  & 216  & 0.00 & 0.00  & 0.00 \\
wi7   & 24500  & 0     & 0    & 0    & \(\infty\) & \(\infty\) & \(\infty\) \\
wi7   & 24500  & 0     & 0    & 0    & \(\infty\) & \(\infty\) & \(\infty\) \\
wi7   & 2450   & 2450  & 343  & 343  & 0.00 & 0.00  & 0.00 \\
wi7   & 4096   & 4096  & 343  & 343  & 0.00 & 0.00  & 0.00 \\
\bottomrule
\end{tabular}%
}
\end{table*}

\begin{table*}[p]
\centering
\caption{Full heavy-hitter data for \(k=4\). The three gap columns give
the percentage increase in the best recovered cost relative to
exhaustive analysis of the raw sample histogram.}
\label{tab:full-hh-k4}

\scriptsize
\setlength{\tabcolsep}{3pt}
\renewcommand{\arraystretch}{0.90}

\resizebox{\textwidth}{!}{%
\begin{tabular}{lrrrrrrr}
\toprule
Instance & Shots
& \(\lvert H_{\mathrm{HH}}\rvert\)
& \(\lvert H_{\mathrm{blk}}\rvert\)
& \(\lvert H_{\mathrm{SR}}\rvert\)
& Gap HH [\%] & Gap blk [\%] & Gap SR [\%] \\
\midrule
dj10  & 50000  & 50000 & 10000 & 2000 & 0.00 & 0.98  & 6.25 \\
dj10  & 5000   & 5000  & 5000  & 2000 & 0.00 & 0.00  & 0.00 \\
dj10  & 4096   & 4096  & 4096  & 2000 & 0.00 & 0.00  & 0.00 \\
dj8   & 32000  & 32000 & 4096  & 2000 & 0.00 & 0.00  & 0.00 \\
dj8   & 32000  & 32000 & 4096  & 2000 & 0.00 & 0.00  & 0.00 \\
dj8   & 32000  & 32000 & 3200  & 2000 & 0.00 & 0.00  & 0.00 \\
dj8   & 4096   & 4096  & 4096  & 2000 & 0.00 & 0.00  & 0.00 \\
dj9   & 40500  & 40500 & 6561  & 2000 & 0.00 & 0.14  & 0.14 \\
dj9   & 40500  & 40500 & 6561  & 2000 & 0.00 & 0.00  & 0.14 \\
dj9   & 4050   & 4050  & 4050  & 2000 & 0.00 & 0.00  & 0.00 \\
dj9   & 4096   & 4096  & 4096  & 2000 & 0.00 & 0.00  & 1.22 \\
eil10 & 50000  & 50000 & 10000 & 2000 & 0.00 & 4.14  & 17.16 \\
eil10 & 5000   & 5000  & 5000  & 2000 & 0.00 & 0.00  & 0.00 \\
eil10 & 4096   & 4096  & 4096  & 2000 & 0.00 & 0.00  & 0.00 \\
eil11 & 60500  & 60500 & 14641 & 2000 & 0.00 & 10.78 & 10.78 \\
eil11 & 6050   & 6050  & 6050  & 2000 & 0.00 & 0.00  & 3.47 \\
eil11 & 4096   & 4096  & 4096  & 2000 & 0.00 & 0.00  & 0.00 \\
wi4   & 8000   & 7543  & 256   & 256  & 0.00 & 0.00  & 0.00 \\
wi4   & 8000   & 7307  & 256   & 256  & 0.00 & 0.00  & 0.00 \\
wi4   & 800    & 793   & 256   & 256  & 0.00 & 0.00  & 0.00 \\
wi4   & 10000  & 242   & 242   & 242  & 0.00 & 0.00  & 0.00 \\
wi4   & 100000 & 253   & 253   & 253  & 0.00 & 0.00  & 0.00 \\
wi4   & 100000 & 246   & 246   & 246  & 0.00 & 0.00  & 0.00 \\
wi4   & 4096   & 3495  & 256   & 256  & 0.00 & 0.00  & 0.00 \\
wi5   & 12500  & 12500 & 625   & 625  & 0.00 & 0.00  & 0.00 \\
wi5   & 12500  & 12499 & 625   & 625  & 0.00 & 0.00  & 0.00 \\
wi5   & 1250   & 1250  & 625   & 625  & 0.00 & 0.00  & 0.00 \\
wi5   & 4096   & 4095  & 625   & 625  & 0.00 & 0.00  & 0.00 \\
wi6   & 18000  & 18000 & 1296  & 1296 & 0.00 & 0.00  & 0.00 \\
wi6   & 18000  & 18000 & 1296  & 1296 & 0.00 & 0.00  & 0.00 \\
wi6   & 1800   & 1800  & 1296  & 1296 & 0.00 & 0.00  & 0.00 \\
wi6   & 4096   & 4096  & 1296  & 1296 & 0.00 & 0.00  & 0.00 \\
wi7   & 24500  & 24500 & 2401  & 2000 & 0.00 & 0.00  & 0.00 \\
wi7   & 24500  & 24500 & 2401  & 2000 & 0.00 & 0.00  & 0.00 \\
wi7   & 2450   & 2450  & 2401  & 2000 & 0.00 & 0.00  & 0.00 \\
wi7   & 4096   & 4096  & 2401  & 2000 & 0.00 & 0.00  & 0.00 \\
\bottomrule
\end{tabular}%
}
\end{table*}

\printbibliography
\end{document}